\documentclass[prx,aps,floatfix,amsmath,amssymb,superscriptaddress,tightenlines,twocolumn]{revtex4}
\usepackage{graphicx}
\usepackage{hyperref}
\usepackage{amsmath}
\usepackage{tikz}
\usepackage{amssymb}

\newcommand{\Rom}[1]{\uppercase\expandafter{\romannumeral #1\relax}}

\usepackage{amsthm}

\newtheorem*{assumption*}{\assumptionnumber}
\providecommand{\assumptionnumber}{}
\makeatletter
\newenvironment{assumption}[2]
{%
	\renewcommand{\assumptionnumber}{Assumption $\mathbf{#1}$}%
	\begin{assumption*}%
		\protected@edef\@currentlabel{$\mathbf{#1}$}%
	}
	{%
	\end{assumption*}
}

\makeatother

\theoremstyle{definition}

\newtheorem*{remark}{Remark}

\theoremstyle{theorem}
\newtheorem{theorem}{Theorem}[section] 

\theoremstyle{corollary}

\theoremstyle{lemma}

\theoremstyle{Proposition} 
\newtheorem{Proposition}[theorem]{Proposition}

\begin{document}
	\title{Seeing topological entanglement through the information convex }
	
	\author{Bowen Shi}
	\affiliation{Department of Physics, The Ohio State University, Columbus, OH 43210, USA}
	
	\date{\today}
	
\begin{abstract}
	The information convex allows us to look into certain information-theoretic constraints in  two-dimensional topological orders. We provide a  derivation of the topological contribution $\ln d_a$ to the von Neumann entropy, where $d_a$ is the quantum dimension of anyon $a$.   This value emerges as the only value consistent with strong subadditivity, assuming a certain topological dependence of the information convex structure. In particular, it is assumed that the fusion multiplicities are coherently encoded in a 2-hole disk. A similar contribution ($\ln d_{\alpha}$) is derived for gapped boundaries. This method further allows us to identify the fusion probabilities and certain constraints on the fusion theory. We also derive a linear bound on the circuit depth of unitary non-Abelian  string operators and discuss how it generalizes and changes in the presence of a gapped boundary. 
\end{abstract}

	\maketitle

\section{Introduction}

Topological orders in two dimensions (2D) \cite{Wen:1989iv,PhysRevB.41.9377} are a long-range entangled \cite{2010PhRvB..82o5138C}  gapped phase of matter whose ground state has topological entanglement entropy (TEE) \cite{2006PhRvL..96k0404K,2006PhRvL..96k0405L}. Its ground-state degeneracy depends on the system topology and these ground states are locally indistinguishable.
Topological orders support \emph{topological excitations}, i.e., anyons in the 2D bulk.
 These topological excitations cannot be created by local operators, but they can be created by  string operators.
When considering excited states with a few topological excitations, the universal topological contribution of von Neumann entropy from each superselection sector can also be identified.
Each \emph{bulk superselection sector} or anyon type $a$  contributes to the von Neumann entanglement of a certain subsystem by $\ln d_a$, where $d_a$ is the quantum dimension of anyon $a$ \cite{2006PhRvL..96k0404K,2008JHEP...05..016D}.

Gapped boundaries may exist in nonchiral topological orders \cite{1998quant.ph.11052B,2011CMaPh.306..663B,2012CMaPh.313..351K,2012arXiv1211.4644K,Levin2013,
	2015PhRvL.114g6402L,2015PhRvL.114g6401H,
	2018JHEP...01..134H,2017PhRvB..96p5138B,2017CMaPh.355..645C,2018JHEP...06..113C,2019PhRvB..99c5112S}. In the presence of a gapped boundary, there are deconfined boundary topological excitations on the 1D gapped boundary. They carry \emph{boundary superselection sectors} and have their own fusion rules. When moving a bulk anyon onto the gapped boundary and measuring its sector, the outcome will be a certain boundary superselection sector. We will refer to this phenomena as bulk-to-boundary condensation. A $\ln d_{\alpha}$ contribution to the von Neumann entropy from each boundary superselection sector $\alpha$ is observed in \cite{2019PhRvB..99c5112S}, where $d_{\alpha}$ is the quantum dimension of $\alpha$.

Given these universal properties of 2D topological orders, it is quite natural  to ask  whether a certain property follows logically from several other properties and whether there is a unified theoretical framework to describe these properties.

 The algebraic theory of anyon is proposed as a universal framework to describe the bulk phase of the topological order \cite{2006AnPhy.321....2K}, and  related proposals \cite{2012CMaPh.313..351K,2012arXiv1211.4644K,Levin2013} for gapped boundaries (of nonchiral topological orders) have also appeared in the literature recently.

A remarkably different line of progress has been made from a quantum information theory perspective. This research brings new insights from the fundamental properties of quantum many-body states. 
It is shown that a nonzero TEE is necessary to have nontrivial ground-state degeneracy \cite{2013PhRvL.111h0503K} and to have nontrivial topological excitations \cite{2015PhRvB..92k5139K}. These results provide strong evidence that entanglement and other universal properties of topologically ordered systems are nontrivially related. 
Focusing on the properties of a quantum state allows us to overcome certain intrinsic difficulties of interacting many-body systems, e.g., the fact that only particular forms of Hamiltonians can be solved for interacting many-body systems.
Moreover, it is shown that a long-range entangled ground state of a topologically ordered system could be converted to a product state only with a quantum circuit whose depth scales at least linearly with the system size~\cite{Haah2016}. 

In a recent work, the information convex  $\Sigma(\Omega)$ has been proposed as a characterization of topological order \cite{2019PhRvB..99c5112S}. It applies to both the bulk and the gapped boundary. $\Sigma(\Omega)$ is a set of density matrices on subsystem $\Omega$; each element is obtained from a density matrix minimizing the Hamiltonian on a subsystem larger than $\Omega$ by many correlation lengths. For topological orders, the information convex is argued to be a low-dimensional convex set with elements locally indistinguishable from the ground state. The structure of the information convex depends on the topology of the subsystem. The information convex  characterizes the fusion and condensation multiplicities in addition to the bulk superselection sectors, the boundary superselection sectors, and their quantum dimensions. It is also a convenient framework to summarize many of the known properties of topological orders.

In this work, we look into certain information-theoretic constraints on the information convex structure.
We show that certain properties of the information convex  follow from  others. In particular, we  derive the $\ln d_a$, $\ln d_{\alpha}$ topological contributions of entropy by showing that they are the only values consistent with the strong subadditivity (SSA) \cite{1973JMP....14.1938L} and a set of assumptions concerning the topology dependence of the information convex structure. 
This method also allows us to calculate the fusion probability of bulk anyons, boundary topological excitations, and the condensation probability from the bulk to a gapped boundary. Also derived are some consistency conditions on the fusion theory. Moreover, we derive a linear bound on the circuit depth of the unitary non-Abelian string operator and discuss how this result generalizes and changes in the presence of a gapped boundary.

The results indicate that the information convex may be treated as a theoretical object having its predictive power from the self-consistency relations, and therefore it is more than a convenient tool to summarize results obtained by other methods. 
We hope this work also extends our toolbox and provide us with another way to explore the fusion rules of loop-like excitations in 3D topological orders and their TEE contributions.

The rest of the paper is organized as follows.
In Sec.~\ref{ground state and excitations}, we briefly review the ground states of topological orders, the topological excitations, and some basics of the information convex. Relevant properties are summarized into a few assumptions. 
In Sec.~\ref{Sec_Bulk}, we describe a key assumption about how the information convex encodes fusion multiplicities and then provide a derivation of the $\ln d_a$ topological contribution to the von Neumann entropy.
In Sec.~\ref{Section:boundary}, we derive a similar contribution ($\ln d_{\alpha}$) from each boundary superselection sector $\alpha$.
In Sec.~\ref{Section:condensation}, we derive a relation about the condensation from bulk to a gapped boundary.
In Sec.~\ref{Section_Info}, we provide some more details of the method which reveals the physical interpretations of several probabilities calculated in earlier sections, and we also discuss the circuit depth of non-Abelian anyon strings.
In Sec.~\ref{Sec:summary}, we conclude with a summary.

\section{The ground states and  excitations of 2D topological orders}\label{ground state and excitations}

\subsection{The ground states}

In this subsection, we summarize relevant properties of topologically ordered ground states into Assumptions \ref{as:G0}, \ref{as:G1}, \ref{as:G2} and \ref{as:G3}. While there are other equivalent ways to write down the assumptions,  we choose to write down the assumptions using \emph{mutual information} and \emph{conditional mutual information}. In this way, it is easier to borrow tools from quantum information theory. We further emphasize that information-theoretic considerations are physically relevant because every physical quantum many-body system obeys the law of quantum mechanics.
Note that a very similar set of assumptions about the topologically ordered ground states has been used in \cite{2016PhRvA..93b2317K}.

In the following, we will use $\rho$, $\sigma$ to denote (reduced) density matrices, and use subscripts $\Omega$, $ABC$  for the subsystems. We will use $S(\sigma)\equiv -\textrm{tr} (\sigma \ln \sigma)$ to denote the von Neumann entropy.  The mutual information is defined as
\begin{equation}
I(A:C)_{\rho}\equiv S(\rho_A) + S(\rho_B) -S(\rho_{AB}) \nonumber
\end{equation}
and the conditional mutual information is defined as
\begin{equation}
I(A:C\vert B)_{\rho}\equiv S(\rho_{AB}) +S(\rho_{BC})-S(\rho_B)-S(\rho_{ABC}). \nonumber
\end{equation}
The strong subadditivity (SSA) \cite{1973JMP....14.1938L}  is the statement that
$I(A:C\vert B)_{\rho}\ge 0$ for any density matrix $\rho_{ABC}$. It is a very powerful and miraculous statement which is known to be the source of nontrivial insights.  As a special case, SSA implies $I(A:C)_{\rho}\ge 0$ for any density matrix  $\rho_{AC}$.
For pedagogical discussions of mutual information, conditional mutual information, and the quantum information theory behind them, we refer the readers to \cite{preskill1998lecture, nielsen2002quantum, 2018arXiv180511965W}.

To be concrete, we assume that the system is on either an infinite plane or a half plane with a single gapped boundary, so that we have only a single ground state $\vert \psi\rangle$, and let us denote the ground state density matrix as
\begin{equation}
\sigma^1\equiv \vert\psi\rangle \langle \psi\vert.
\end{equation}
The method could be adapted to other manifolds.
 The ground state is gapped; we neglect the correlations at a large enough length scale $\epsilon$. We will talk about the topology of a subsystem, and it is understood that the topology is well defined only at a length scale above $\epsilon$.

\begin{figure}[h]
	\centering
	\includegraphics[scale=1.0]{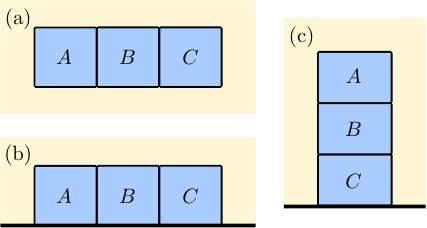}
	\caption{(a) The bulk subsystems $A$, $B$, $C$ considered in assumption \ref{as:G1}. (b) The  subsystems  $A$, $B$, $C$  attached to a gapped boundary considered in assumption \ref{as:G2}. (c) The  subsystems  $A$, $B$, $C$  considered in assumption \ref{as:G3}.}
	\label{ABC_Trivial}
\end{figure}

\begin{assumption}{G_0}{}\label{as:G0}
	Consider the ground state of a topological order on an infinite plane or a half plane bounded by a gapped boundary. For arbitrary subsystems $A$ and  $C$ which are separated by a distance greater than some length scale $\epsilon$, the mutual information vanishes:
	\begin{equation}
	I(A:C)_{\sigma^1}=0. \nonumber \label{eq: Assumption_G0}
	\end{equation}
	Here $A$, $C$  can be either in the bulk or touch  the boundary.
\end{assumption}

\begin{remark}
	It is possible to generalize \ref{as:G0} to other system topologies. However, one should be aware of potential counterexamples. For example, on  a generic  ground state of a torus $T^2$, it is possible to pick two annuli $A$ and $C$ which are separated by a large distance and $I(A:C)> 0$; see \cite{2015arXiv150807006J}. In this case, each  annulus could not shrink to a point by continuous deformation on $T^2$. Also notice that if \ref{as:G0} is satisfied, and $\tilde{\sigma}^1$ is related to $\sigma^1$ by a finite-depth quantum circuit, then $I(A:C)_{\tilde{\sigma}^1}=0$ when $A$ and $C$ are separated by a length scale $\epsilon'$ equal to $\epsilon$ plus circuit depth. 
\end{remark}

The following three assumptions, Assumptions  \ref{as:G1}, \ref{as:G2} and \ref{as:G3}, are statements about the vanishing of conditional mutual information for certain subsystem choices.

\begin{assumption}{G_1}{} \label{as:G1}
	For any subsystem choice ABC topologically equivalent to the one shown in Fig.~\ref{ABC_Trivial}(a),
	\begin{equation}
	I(A:C\vert B)_{\sigma^1}=0.\nonumber
	\end{equation}
\end{assumption}

\begin{assumption}{G_2}{} \label{as:G2}
	For any subsystem choice ABC topologically equivalent to the one shown in Fig.~\ref{ABC_Trivial}(b),
	\begin{equation}
	I(A:C\vert B)_{\sigma^1}=0.\nonumber
	\end{equation}
\end{assumption}

\begin{assumption}{G_3}{} \label{as:G3}
	For any subsystem choice ABC topologically equivalent to the one shown in Fig.~\ref{ABC_Trivial}(c),
	\begin{equation}
	I(A:C\vert B)_{\sigma^1}=0.\nonumber
	\end{equation}
\end{assumption}

\begin{remark}
	Assumptions \ref{as:G1}, \ref{as:G2} and \ref{as:G3} may be understood in the same manner; i.e., the von Neumann entropy of a subsystem (of a topologically ordered ground state) can be separated into local contributions plus a universal contribution depending on the subsystem topology. For all the cases in Fig.~\ref{ABC_Trivial}, both the local contributions and the topological contributions cancel. Note, however, that unlike \ref{as:G0}, the vanishing of conditional mutual information in \ref{as:G1}, \ref{as:G2} and \ref{as:G3} may break down under finite-depth quantum circuits \cite{2016PhRvB..94g5151Z,2019PhRvL.122n0506W}. Nonetheless, all of the examples involve some sort of symmetry. It is conjectured that for a generic quantum circuit without any symmetry,  \ref{as:G1}, \ref{as:G2} and \ref{as:G3} still hold, albeit the quantum circuit  changes the scale $\epsilon$.
\end{remark}

\subsection{The excitations}

We describe three assumptions, Assumptions \ref{as:S1}, \ref{as:S2}, \ref{as:S3} about the unitary string operators. The unitary string operators, when acting upon the ground state, can create deconfined topological excitations inside the bulk or along a gapped boundary. To state the assumptions, we first need to review the superselection sectors briefly.

We use $a$, $b$  to label \emph{bulk superselection sectors} or bulk anyon types \cite{2006AnPhy.321....2K}.  Each bulk superselection sector is an equivalent class of deconfined bulk excitations, and the excitations (either a single excitation or several excitations) belonging to one class cannot be transformed into another class by any local operator.
We use $\alpha$, $\beta$  to label \emph{boundary superselection sectors} or boundary topological excitation types. Each boundary superselection sector is an equivalent class of deconfined boundary excitations \cite{2012CMaPh.313..351K}, and the boundary excitations (either a single excitation or several excitations  lying along the boundary) belonging to one class cannot be transformed into another class by any local operator acting around the boundary.

It should be noted that not every pair of excitations can be connected by a string of the type in Fig.~\ref{Unitary_String_Assumption}.
An anyon $a$ is always connected with its antiparticle $\bar{a}$ for the string operator in Fig.~\ref{Unitary_String_Assumption}(a); a boundary topological excitation $\alpha$ is always connected with its antiparticle $\bar{\alpha}$ for the string operator in Fig.~\ref{Unitary_String_Assumption}(b); $\alpha$ must be in the condensation channel of $a$ if there is a string of the type shown in Fig.~\ref{Unitary_String_Assumption}(c) which connects $a$ and $\bar{\alpha}$. We summarize these into Assumptions \ref{as:S1}, \ref{as:S2} and \ref{as:S3} below.

\begin{figure}[h]
	\centering
	\includegraphics[scale=1.0]{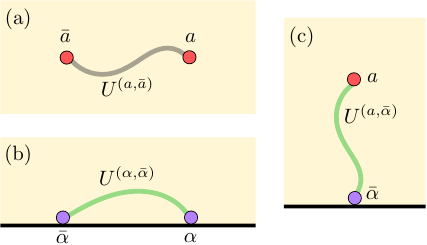}
	\caption{(a) A unitary string operator $U^{(a,\bar{a})}$ in the bulk, which creates anyon pair  $(a,\bar{a})$. (b) A unitary string operator $U^{(\alpha,\bar{\alpha})}$  which creates boundary topological excitation pair  $(\alpha,\bar{\alpha})$. (c) A unitary string operator $U^{(a,\bar{\alpha})}$ which creates a pair of topological excitations: a bulk anyon $a$ and a boundary topological excitation $\bar{\alpha}$ which is in the condensation channel of $\bar{a}$. The choice of color:  Gray strings are in the bulk and green strings touch the boundary; red dots are bulk anyons and purple dots are boundary topological excitations. }
	\label{Unitary_String_Assumption}
\end{figure}

\begin{assumption}{S_1}{} \label{as:S1}
	A pair of bulk anyons $(a,\bar{a})$, shown in  Fig.~\ref{Unitary_String_Assumption}(a), can be created by applying a unitary string operator $U^{(a,\bar{a})}$ onto the ground state. The string is inside the bulk and the anyons $a$ and $\bar{a}$ are localized in a small area at the two ends of the string. The anyons are deconfined, and the middle part of the string does not increase the energy density. 
\end{assumption}

\begin{assumption}{S_2}{} \label{as:S2}
	A pair of boundary topological excitations $(\alpha,\bar{\alpha})$ along the same gapped boundary, shown in Fig.~\ref{Unitary_String_Assumption}(b), can be created by applying a unitary string operator $U^{(\alpha,\bar{\alpha})}$ onto the ground state. The boundary topological excitations $\alpha$ and $\bar{\alpha}$ are localized in a small area at the two ends of the string. 
	The support of $U^{(\alpha,\bar{\alpha})}$ touches the boundary since the two ends are along the boundary. On the other hand, the middle part can either lie along the boundary or stretch into the bulk. The boundary topological excitations  are deconfined, and the middle part of the string does not increase the energy density.
\end{assumption}

\begin{assumption}{S_3}{} \label{as:S3}
	A pair of topological excitations $(a,\bar{\alpha})$, shown in Fig.~\ref{Unitary_String_Assumption}(c), can be created by applying a unitary string operator $U^{(a,\bar{\alpha})}$ onto the ground state if $\alpha$ is in the condensation channel of $a$. The bulk anyon $a$ and boundary topological excitations  $\bar{\alpha}$ are localized in a small area at the two ends of the string. The  topological excitations  are deconfined and the middle part of the string does not increase the energy density. 
\end{assumption}

\begin{remark}
$\,$
\begin{enumerate}
	\item Nonunitary string operators are discussed in many references. They naturally appear in exactly solvable models with non-Abelian anyons \cite{2003AnPhy.303....2K,2005PhRvB..71d5110L,2008PhRvB..78k5421B}. How do these operators relate to our unitary string operators? In fact, it follows from \ref{as:G0} that a unitary string operator exists for every nonunitary string operator, albeit the support of the unitary string operator is usually slightly thicker. A short proof is presented in Appendix \ref{appendix:string}.
	
	\item Implicitly, we have assumed that we have the freedom to choose the support of the string; i.e., the string can be deformed. We are aware that this deformation assumption has been used in showing that any system that supports anyons must have a nonvanishing topological entanglement entropy \cite{2015PhRvB..92k5139K}. 
	
	\item \ref{as:S1} says that a pair of anyons $(a,\bar{a})$ could be created applying a unitary bulk string $U^{(a,\bar{a})}$. It does not imply that any state with two anyons $(a,\bar{a})$ could be obtained in this way. A string attached to the boundary, see Fig.~\ref{Unitary_String_Depth}(b), may prepare a different quantum state with $(a,\bar{a})$ lying in the same positions.
	
    \item We aim to give clear statements about the assumptions; however, further reducing the set of assumptions requires additional work. These assumptions have been recently shown to emerge in a general theoretical framework based on an entanglement area law~\cite{2019arXiv190609376S}.
	 
	\item For non-Abelian $a$ and $\alpha$, we will discuss a linear bound on the circuit depth of $U^{(a,\bar{a})}$ and $U^{(\alpha,\bar{\alpha})}$, see Sec.~\ref{Sec._Depth}.
\end{enumerate}
\end{remark}

	\subsection{A brief review of the information convex}
In this subsection, we briefly review the notion of the \emph{information convex} introduced in \cite{2019PhRvB..99c5112S}. 
The information convex  $\Sigma(\Omega)$ is a set of density matrices on subsystem $\Omega$, each element is obtained from a density matrix minimizing the Hamiltonian on a slightly larger subsystem $\Omega'$ (which is bigger than $\Omega$ at least by the length scale $\epsilon$ in \ref{as:G0}). The original definition applies to frustration-free local Hamiltonians, and concrete calculations have been done in quantum double models \cite{2003AnPhy.303....2K,2008PhRvB..78k5421B,2011CMaPh.306..663B}. In light of the calculation, the following results are expected to hold more generally.

For topological orders, the information convex $\Sigma(\Omega)$ is a low-dimensional but nontrivial convex set, the structure of which depends on the topology of subsystem $\Omega$. It  
captures how topological excitations modify the density matrix on a subsystem away from them and how diverse the density matrices can be.
The information convex characterizes the fusion and condensation multiplicities in addition to the bulk superselection sectors $\{a \}$, boundary superselection sectors $\{ \alpha \}$ and their quantum dimensions $\{d_a\}$ and $\{d_{\alpha} \}$. The information convex also provides a convenient framework to summarize many of the previously known information-theoretic properties of topological orders.

Given the evidence that the information convex may provide a useful concept and that explicit results are only available to some particular models up to now, an  outstanding problem is to establish its theoretical foundation with more general principles. The assumptions listed below and the results derived from them should be interpreted as some initial progress toward this goal, and these assumptions should not be interpreted as the axioms in the final theoretical framework.

Also, note that the original definition with the frustration-free local Hamiltonian is not crucial here since the results in this paper can be derived as long as the assumptions hold.

\begin{figure}[h]
	\centering
	\includegraphics[scale=1.0]{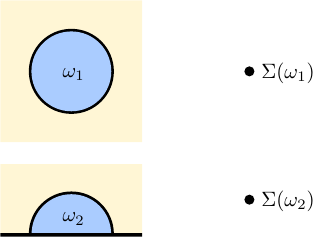}
	\caption{Subsystems $\omega_1$, $\omega_2$ and the corresponding information convex. $\Sigma(\omega_1)$  contains a single element $\sigma^1_{\omega_1}$ and  $\Sigma(\omega_2)$  contains a single element $\sigma^1_{\omega_2}$.}
	\label{figure_omega_1_2}
\end{figure}

We will consider subsystems $\omega_1$, $\omega_2$, $\Omega_1$, $\Omega_2$, $\Omega_3$, $\Omega_4$, $\Omega_5$, $\Omega_6$ in this paper. Each of them is a label of topology, and the relation to the boundary is considered as part of the topological data. Let us start with the simplest topologies $\omega_1$ and $\omega_2$, shown in Fig.~\ref{figure_omega_1_2}. 
(Recall $\vert\psi\rangle$ is the ground state, and we will use $\bar{A}$ to denote the complement of subsystem $A$.)

\begin{assumption}{\omega_1}{} \label{as:1}
	For a disk $\omega_1$ in the bulk, see Fig.~\ref{figure_omega_1_2},
	\begin{equation}
	\Sigma(\omega_1)=\{ \sigma^1_{\omega_1} \},
	\end{equation} 
	where $\sigma^1_{\omega_1} \equiv \textrm{tr}_{\bar{\omega}_1}\vert\psi\rangle \langle \psi\vert$. 
\end{assumption}

\begin{assumption}{\omega_2}{} \label{as:1'}
	For a disk $\omega_2$ attached to the gapped boundary on a single connected component, see Fig.~\ref{figure_omega_1_2},
	\begin{equation}
	\Sigma(\omega_2)=\{ \sigma^1_{\omega_2} \},
	\end{equation} 
	where $\sigma^1_{\omega_2} \equiv \textrm{tr}_{\bar{\omega}_2}\vert\psi\rangle \langle \psi\vert$.
\end{assumption}

\begin{remark}
	Assumption \ref{as:1} is a convenient  way to say that a disk cannot tell whether there are excitations in other places. If we apply \ref{as:1} to a manifold with multiple ground states, one recovers the well-known fact that all ground states are indistinguishable on a disk subsystem.  The condition \ref{as:1} is also known as the TQO-2 condition \cite{bravyi2010topological}, which is crucial in the study of perturbations in topological orders. \ref{as:1'} is simply a natural generalization of \ref{as:1} to a system with gapped boundaries.
\end{remark}

\begin{figure}[h]
	\centering
	\includegraphics[scale=1.0]{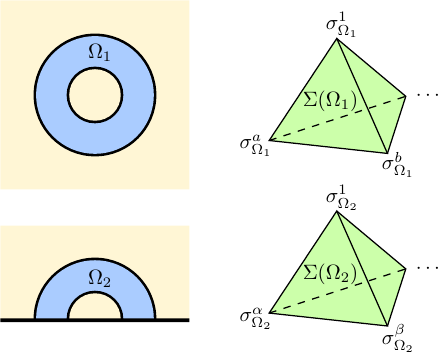}
	\caption{Bulk annulus $\Omega_1$ and half annulus $\Omega_2$ attach to the boundary  and the corresponding information convex. 
	 $\Sigma(\Omega_1)$ and $\Sigma(\Omega_2)$ are always simplex but usually not tetrahedrons. Here, ``$\cdots$" represents potentially more extreme points that are omitted.
	}\label{figure_Omega_1_2}
\end{figure}

The next two topologies $\Omega_1$ and $\Omega_2$ are relevant to superselection sectors; see Fig.~\ref{figure_Omega_1_2}.

\begin{assumption}{\Omega_1}{} \label{as:2}
	For a bulk annulus $\Omega_1$, see Fig.~\ref{figure_Omega_1_2}, 
	\begin{equation}
	\Sigma(\Omega_1) = \left\{ \sigma_{\Omega_1} \left\vert  \sigma_{\Omega_1}=\sum_{a} p_a \sigma^a_{\Omega_1}\right. \right \},
	\end{equation}
	where $a$ is the label for bulk superselection sectors (or anyon types).
	 $\{p_a\}$ is a probability distribution and $\sigma^a_{\Omega_1}$ is an extreme point. 
	 \begin{enumerate}
	 	\item Distinct extreme points are orthogonal:
	 	\begin{equation}
	 	\sigma^a_{\Omega_1}\cdot \sigma^b_{\Omega_1}=0, \quad \forall a\ne b. \label{eq:othogonal_a} 
	 	\end{equation}
	 	\item There is a universal contribution to von Neumann entropies:
	 	\begin{equation}
	 	S(\sigma^a_{\Omega_1})=S(\sigma^1_{\Omega_1}) + 2 f(a),
	 	\end{equation}
    	where $1$ is the vacuum sector; $f(a)$ is real, $f(1)=0$, and $f(a)=f(\bar{a})$; $\bar{a}$ is the antiparticle of $a$.
	    \item The extreme point $\sigma^a_{\Omega_1}$ is obtained on an annulus surrounding an anyon $a$.
	\end{enumerate}
\end{assumption}

\begin{assumption}{\Omega_2}{} \label{as:2'}
	For a half annulus $\Omega_2$ attaches to the boundary,
	 see Fig.~\ref{figure_Omega_1_2}, 
	\begin{equation}
	\Sigma(\Omega_2) = \left\{ \sigma_{\Omega_2} \left\vert  \sigma_{\Omega_2}=\sum_{\alpha} p_{\alpha} \sigma^{\alpha}_{\Omega_2} \right. \right\},
	\end{equation}
	where $\alpha$ is the label for boundary superselection sectors (or deconfined boundary topological excitations types). 
	 $\{p_{\alpha}\}$ is a probability distribution and $\sigma^{\alpha}_{\Omega_2}$ is an extreme point. 
	 \begin{enumerate}
	 	\item Distinct extreme points are orthogonal:
	 	\begin{equation}
	 	\sigma^{\alpha}_{\Omega_2}\cdot \sigma^{\beta}_{\Omega_2}=0, \quad \forall \alpha \ne \beta . \label{eq:othogonal_b}
	 	\end{equation}
	 	\item There is a universal contribution to von Neumann entropies:
	 	\begin{equation}
	 	S(\sigma^{\alpha}_{\Omega_2}) = S(\sigma^1_{\Omega_2}) + 2 F(\alpha), \label{eq:F_alpha}
	 	\end{equation}
	 	where $1$ is the vacuum sector; $F(\alpha)$ is real, $F(1)=0$, and $F(\alpha)=F(\bar{\alpha})$; $\bar{\alpha}$ is the antiparticle of $\alpha$.
	 	\item The extreme point $\sigma^{\alpha}_{\Omega_2}$ is obtained on an $\Omega_2$ surrounding a boundary topological excitation $\alpha$.
	 \end{enumerate}
\end{assumption}

\begin{remark}
Assumptions \ref{as:2} and \ref{as:2'}   summarize how $\Omega_1$ and $\Omega_2$ detect each (bulk or boundary) superselection sector and how each superselection sector contributes to the von Neumann entropy. Intuitively, the string operators $U^{(a,\bar{a})}$ ($U^{(\alpha,\bar{\alpha})}$) must pass through the subsystem $\Omega_1$ ($\Omega_2$) and this is why the density matrices are modified. Moreover, we argue that the orthogonal relation in Eq. (\ref{eq:othogonal_a}) is a many-body effect. For large system sizes,  two distinct extreme points must be (to a very good approximation) orthogonal to each other as long as (1) they are short-range entangled when viewed as a 1D system (the 1D is the radial direction), and (2) they have a finite amount of difference on each thin annulus shell. A similar argument could explain the orthogonal relation Eq. (\ref{eq:othogonal_b}). 
\end{remark}

The universal contributions to von Neumann entropies, i.e., $f(a)$ and $F(\alpha)$, are assumed to be generic real numbers. We will later derive their values  $f(a)=\ln d_a$, see Eq. (\ref{eq:d_a}), and $F(\alpha) =\ln d_{\alpha}$, see Eq. (\ref{eq:d_alpha}). They consist of the main results of this paper. Despite that these values have been obtained from other methods previously, e.g., \cite{2006PhRvL..96k0404K,2008JHEP...05..016D,2019PhRvB..99c5112S}, our derivation points to a different perspective on the origin of these numbers because the assumptions are different. In particular, we assume neither an underlying field theory nor an exactly solvable lattice model. Rather the values emerge from information-theoretic constraints.

We would also like to compare the topological contribution $\ln d_a$ in the  $S(\sigma^a_{\Omega_1})$ and the same  topological contribution  for a disk containing an anyon. It is known that a disk containing a single anyon will have von Neumann entropy bigger than the ground state von Neumann entropy by $\ln d_a$. Note, however, that this result is for an anyon pinned down to a fixed position. In general, if we allow the anyon on the disk to entangle its position with an anyon on the rest of the system, 
the entanglement entropy can further grow. Intuitively, the topological entanglement and ``particle entanglement" are added up. On the other hand, when we consider the von Neumann entropy of an element $\sigma_{\Omega}\in \Sigma(\Omega)$, there are no particles inside $\Omega$, and we neatly pick out the topological contributions even if the excitations on the rest of the system have ``particle entanglement".

We will give explicit statements about $\Sigma(\Omega_4)$, $\Sigma(\Omega_5)$ and $\Sigma(\Omega_6)$ in later sections, i.e., Assumptions \ref{as:3}, \ref{as:3'} and \ref{as:1''}. These subsystems have the feature that the fusion multiplicities $N_{ab}^c$, ${N}_{\alpha\beta}^{\gamma}$ (or condensation multiplicities $N_{a}^{\alpha}$) manifest in the structure of extreme points.
This is crucial for our method to work because fusion multiplicities contain enough information to derive the quantum dimensions $\{d_a\}$ (and $\{ d_{\alpha} \}$).  Each multiplicity greater than 1 contributes a set of extreme points parametrized by continuous variables. The interested reader may take a quick look at Fig.~\ref{Bulk_Subsystems}, \ref{Boundary_Subsystems} and \ref{Bulk_Boundary_Subsystems}.

\section{The  bulk of a 2D topological order} \label{Sec_Bulk}

For the 2D bulk of a topological order, the superselection sectors correspond to the anyon types (the vacuum sector is included) \cite{2006AnPhy.321....2K}. We will use the terminology bulk superselection sectors, anyon types, and topological charges interchangeably.
Let us label the bulk superselection sectors by $ \{ 1,a,b,c,\cdots \}$, where $1$ is the vacuum sector. The fusion rule of anyons can be written as
\begin{equation}
a\times b = \sum_c N_{ab}^{c} \,c . \label{eq:bulk fusion rule}
\end{equation}
Here $\{N_{ab}^c\}$ is the set of fusion multiplicities. They are non-negative integers satisfying $N_{ab}^c=N_{ba}^c$ and the set of conditions in Appendix~\ref{Appendix_Fusion}. 

The goal of this section is to provide a derivation of the $\ln d_a$ universal topological contribution of von Neumann entropy given the fusion multiplicities $\{ N_{ab}^c \}$. In particular, we will show that the $\ln d_a$ contribution emerges as the only value consistent with SSA and a set of assumptions about the information convex. Similar methods are applied to study boundary topological excitations along a gapped boundary (see Sec.~\ref{Section:boundary}) and bulk-to-boundary condensation (see Sec.~\ref{Section:condensation}).
A few other implications will be explored after a more detailed discussion in Sec.~\ref{Section_Info}.

\subsection{Preparing for the derivation}\label{Bulk_Assumptions}
In this subsection, we state and explain the key Assumption \ref{as:3} which describes the way $N_{ab}^c$ encoded in $\Sigma(\Omega_4)$. We further state Proposition~\ref{as:4} about the existence of a certain element which saturates SSA. The proof of proposition \ref{as:4} is provided in Sec.~\ref{Section_Info}. Both of them are crucial in the derivation of the $\ln d_a$ contribution to the von Neumann entropy.

\begin{figure}[h]
	\centering
	\includegraphics[scale=1.0]{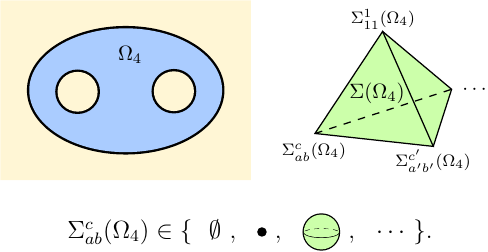}
	\caption{ The bulk subsystem $\Omega_4$, i.e., a 2-hole disk, and its information convex $\Sigma(\Omega_4)$.
		 $\Sigma(\Omega_4)$ is the set of convex combinations of $\Sigma_{ab}^c(\Omega_4)$. The geometry of each $\Sigma_{ab}^c(\Omega_4)$ depends solely on a non-negative integer $N_{ab}^c$. For example, when $N_{ab}^c=0,1$ and $2$, $\Sigma_{ab}^c(\Omega_4)$ is isomorphic to an empty set, a point, and a solid ball, respectively. }
	\label{Bulk_Subsystems}
\end{figure}

\begin{figure}[h]
	\centering
	\includegraphics[scale=1.0]{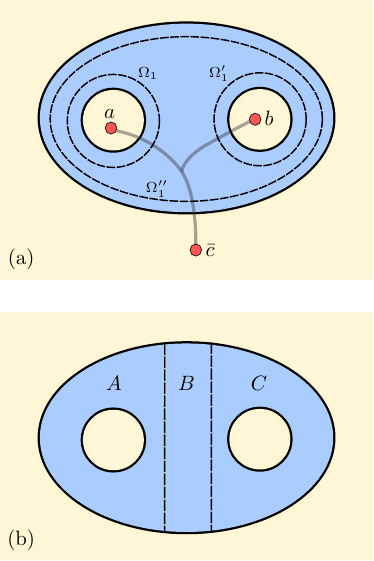}
	\caption{(a) Annuli $\Omega_1, \Omega'_1, \Omega''_1 \subseteq \Omega_4$ surrounding the three entanglement cuts could detect the topological charges.  In this picture, we have $\sigma_{\Omega_1}^a$, $\sigma^b_{\Omega'_1}$, and $\sigma^c_{\Omega''_1}$ on these subsystems. (b) $\Omega_4$ is divided into smaller pieces $\Omega_4=ABC$. Both $AB$ and $BC$ have $\Omega_1$ topology.}
	\label{Bulk_Restriction}
\end{figure}

\begin{assumption}{\Omega_4}{}  \label{as:3}
	
		For a subsystem $\Omega_4$, i.e., a 2-hole disk in the bulk, see Fig.~\ref{Bulk_Subsystems}, the information convex $\Sigma(\Omega_4)$ has the following structure:
	
	The set of extreme points of $\Sigma(\Omega_4)$ forms a set $\mathcal{M}=\bigcup_{(a,b,c)} \mathcal{M}^{c}_{ab}(\Omega_4)$. Each $\mathcal{M}^{c}_{ab}(\Omega_4)$ with $N_{ab}^c \ne 0$ is a connected component of $\mathcal{M}$. Let $\Sigma^{c}_{ab}(\Omega_4)$ be the convex subset of $\Sigma(\Omega_4)$  formed by a convex combination of elements in $\mathcal{M}^{c}_{ab}(\Omega_4)$. 
	
	Taking a partial trace to reduce  $\sigma_{\Omega_4}^{(a,b,c)_x}\in \Sigma^{c}_{ab}(\Omega_4) $ to certain subsystems of $\Omega_1$ topology surrounding each entanglement cut, i.e., the $\Omega_1$, $\Omega'_1$ and $\Omega''_1$ in Fig.~\ref{Bulk_Restriction}(a), we get extreme points  $\sigma^a_{\Omega_1}$, $\sigma^b_{\Omega'_1}$ and $\sigma^c_{\Omega''_1}$, respectively.
	
	The density matrix $\sigma_{\Omega_4}^{(a,b,c)_x}\in \Sigma^{c}_{ab}(\Omega_4)$ has a one-to-one correspondence with $\hat{\rho}^{(a,b,c)_x}$, a  density matrix on the fusion Hilbert space $\mathbb{V}_{ab}^c\equiv span \{ \vert i\rangle\}_{i=1}^{N_{ab}^c} $. Extreme points of $\Sigma^{c}_{ab}(\Omega_4)$ correspond to the pure-state density matrices on $\mathbb{V}_{ab}^c$.  The mapping preserves the convex structure, i.e.,
	\begin{equation}
	p\,\sigma_{\Omega_4}^{(a,b,c)_x} + (1-p) \sigma_{\Omega_4}^{(a,b,c)_y} =\sigma_{\Omega_4}^{(a,b,c)_z}  \nonumber
	\end{equation}
	if and only if 
	\begin{equation}
	p\,\hat{\rho}^{(a,b,c)_x} + (1-p) \hat{\rho}^{(a,b,c)_y} =\hat{\rho}^{(a,b,c)_z}, \nonumber
	\end{equation}
	where $p\in [0,1]$.
	Furthermore, $\sigma_{\Omega_4}^{(a,b,c)_x}$ has von Neumann entropy: 
	\begin{equation}
	\begin{aligned}
	S(\sigma_{\Omega_4}^{(a,b,c)_x}) =& S(\sigma_{\Omega_4}^{1}) + f(a) + f(b) + f(c)\\
	&\,+ S(\hat{\rho}^{(a,b,c)_x}). \label{Eq_Entanglement_abc}
	\end{aligned}
	\end{equation}
	Here $\sigma^1_{\Omega_4}\equiv \textrm{tr}_{\bar{\Omega}_4}\,\sigma^1$ is the unique element of $\Sigma^{1}_{11}(\Omega_4)$.
\end{assumption}

\begin{remark}
	Assumption \ref{as:3} is the statement that fusion multiplicities $N_{ab}^c$ are coherently encoded in the information convex  $\Sigma(\Omega_4)$.
	Furthermore, \ref{as:3} says that different vectors of the fusion Hilbert space could be determined (up to an overall phase factor) by looking at the element $\sigma^{(a,b,c)_x}_{\Omega_4}\in \Sigma^{c}_{ab}(\Omega_4)$.
\end{remark}

 \begin{Proposition}  \label{as:4}
 	There exists a unique element $\sigma^{(a,b)}_{\Omega_4} \in \Sigma(\Omega_4)$ such that for the partition $\Omega_4=ABC$ shown in Fig.~\ref{Bulk_Restriction}(b) the following hold:
	\begin{enumerate}
		\item It prepares the extreme point $\sigma_{AB}^a$ when restricted to $AB$ and
		it prepares the extreme point $\sigma_{BC}^{b}$ when restricted to $BC$. 
		\item It saturates the conditional mutual information:
		\begin{equation}
		I(A:C\vert B)_{\sigma^{(a,b)}}=0. \nonumber
		\end{equation}
	\end{enumerate}
\end{Proposition}
The proof of Proposition~\ref{as:4} will be presented in Sec.~\ref{Section_Info}. For now, we simply point out that the proof needs \ref{as:G1}, \ref{as:S1}, \ref{as:2} and SSA.

\subsection{ The derivation of the $\ln d_a$ contribution to von Neumann entropy}\label{Sec_Bulk_da}

We present a derivation of the topological contribution to the von Neumann entropy on annulus $\Omega_1$ by showing $f(a)=\ln d_a$, see Eq.~(\ref{eq:d_a}), and the probability in Eq.~(\ref{eq:Bulk Fusion Prob}) from Assumptions \ref{as:G1}, \ref{as:S1}, \ref{as:1}, \ref{as:2} and \ref{as:3}.

First, we apply Proposition~\ref{as:4} to the $(a,b)$ sector and the vacuum  sector $(1,1)$, to express the von Neumann entropy $S(\sigma^{(a,b)}_{\Omega_4})$ and $S(\sigma^1_{\Omega_4})$ in terms of entropies on simpler subsystems $AB$, $BC$ and $B$ in Fig.~\ref{Bulk_Restriction}(b).  We then use  \ref{as:2} and \ref{as:1} to compare the von Neumann entropy in these simpler subsystems and find
\begin{equation}
S(\sigma^{(a,b)}_{\Omega_4})=S(\sigma^{1}_{\Omega_4}) + 2f(a) + 2f(b). \label{star_2}
\end{equation}
Note that $\sigma^{(1,1)}_{\Omega_4}= \sigma^1_{\Omega_4}$ follows from \ref{as:3} and proposition~\ref{as:4}.

Next, we note that  $\sigma^{(a,b)}_{\Omega_4}$ defined in Proposition~\ref{as:4} must be an element with maximal entropy among the elements of $\Sigma(\Omega_4)$ which have topological charge $a$ and $b$ on the two entanglement cuts surrounded by $\Omega_1$ and $\Omega'_1$ in Fig.~\ref{Bulk_Restriction}(a). This follows from SSA,  \ref{as:1}, and \ref{as:2}.
 With \ref{as:3}, we express $\sigma^{(a,b)}_{\Omega_4}$ as a convex combination of elements in $\bigcup_{c}\Sigma_{ab}^c(\Omega_4)$, and calculate its entropy with the help of \ref{as:2}. By maximizing the von Neumann entropy, we find
\begin{equation}
S(\sigma_{\Omega_4}^{(a,b)}) = S(\sigma^1_{\Omega_4}) +f(a) + f(b) + \ln (\sum_c N_{ab}^c\, e^{f(c)}), \label{star_1}
\end{equation}
and that
\begin{equation}
\sigma^{(a,b)}_{\Omega_4} = \sum_{c} P_{(a\times b\to c)} \sigma_{\Omega_4}^{{(a,b,c)}_{\textrm{max}} }, \label{eq:sigma(ab)}
\end{equation}
where $\sigma_{\Omega_4}^{(a,b,c)_{\textrm{max}}}$ is the maximal-entropy element of $\Sigma^{c}_{ab}(\Omega_4)$ and
$P_{(a\times b\to c)} =\frac{N_{ab}^{c}e^{f(c)}}{\sum_d N_{ab}^{d}e^{f(d)}}$. 

In more detail, it follows from \ref{as:2} and \ref{as:3} that
$\sigma_{\Omega_4}^{(a,b,c)_x}\cdot \sigma_{\Omega_4}^{(a',b',c')_y}=0$  if $(a,b,c)\ne (a',b',c')$, where $\sigma_{\Omega_4}^{(a,b,c)_x}\in\Sigma_{ab}^c(\Omega_4)$ and $\sigma_{\Omega_4}^{(a',b',c')_y}\in \Sigma_{a'b'}^{c'}(\Omega_4)$.
This is because of \ref{as:2} and the general result that $\rho_{AB}\cdot\sigma_{AB}=0$ if $\rho_{A}\cdot \sigma_A=0$. (See Appendix \ref{Appendix_Monotonicity} for a proof.) 
Moreover, $S(\sigma_{\Omega_4}^{(a,b,c)_{x}}) \le S(\sigma^1_{\Omega_4}) +f(a) +f(b) + f(c) +\ln N_{ab}^c$ for any $\sigma_{\Omega_4}^{(a,b,c)_{x}}\in \Sigma_{ab}^c(\Omega_4)$. The unique element that saturates the bound is $\sigma_{\Omega_4}^{(a,b,c)_{\textrm{max}}}$.
In the calculation of $S(\sigma_{\Omega_4}^{(a,b)})$ in Eq.~(\ref{star_1}), we have also used the well-known result
$S(\sum_{i} p_i \,\rho^i)=\sum_{i} p_i (S(\rho^i)-\ln p_i)$ 
if  $\rho^i\cdot \rho^j=0,\,\,\forall i\ne j$, and $\{p_i \}$ is a probability distribution.

We have obtained two expressions of $S(\sigma_{\Omega_4}^{(a,b)})$ in Eq.~(\ref{star_2}) and  Eq.~(\ref{star_1}). By comparing them, one finds 
\begin{equation}
e^{f(a)} e^{f(b)}= \sum_c N_{ab}^c e^{f(c)}. \label{eq:e^f(a)}
\end{equation}
Because $f(a)$ is real, we must have $e^{f(a)}\in (0,+\infty)$. 
 Equation~(\ref{eq:e^f(a)}) has a unique solution $e^{f(a)}=d_a$, where $d_a$ is the quantum dimension uniquely defined given $N_{ab}^c$ \cite{2006AnPhy.321....2K}. See also Appendix \ref{Appendix_Fusion} for a self-contained proof using a few assumptions about the fusion multiplicities. Thus,
\begin{equation}
f(a)=\ln d_a \label{eq:d_a}.
\end{equation}

To summarize, we have derived the $\ln d_a$ topological contribution of von Neumann entropy as the only value consistent with SSA given Assumptions \ref{as:G1}, \ref{as:S1}, \ref{as:1}, \ref{as:2} and \ref{as:3}. We have also derived the explicit form of density matrix  $\sigma^{(a,b)}_{\Omega_4}$ in Eq.~(\ref{eq:sigma(ab)}) with probability
\begin{equation}
P_{(a\times b\to c)} =\frac{N_{ab}^c d_c}{d_a d_b}. \label{eq:Bulk Fusion Prob}
\end{equation} 
Later  we will identify the physical meaning of probability $P_{(a\times b\to c)}$; see Sec.~\ref{Sec_Fusion probability}.

\section{The gapped boundary of a 2D non-chiral topological order} \label{Section:boundary}
Two-dimensional nonchiral topological orders may have gapped boundaries \cite{1998quant.ph.11052B,2011CMaPh.306..663B,2012CMaPh.313..351K,2012arXiv1211.4644K,Levin2013,2015PhRvL.114g6402L,2015PhRvL.114g6401H,
	2018JHEP...01..134H,2017PhRvB..96p5138B,2017CMaPh.355..645C}. A bulk phase may have more than one gapped boundary type.
 For each boundary type, there are several \emph{boundary superselection sectors} \cite{2012CMaPh.313..351K,2012arXiv1211.4644K,2019PhRvB..99c5112S}, i.e., the types of deconfined \emph{boundary topological excitations}. We denote the boundary superselection sectors using $\{ 1, \alpha,\beta,\gamma, \cdots \}$, where $1$ is the boundary vacuum. 
 Note that the boundary superselection sectors are in general different from bulk superselection sectors, and for non-Abelian models, they usually cannot be identified as a subset of bulk superselection sectors.  Boundary topological excitations can fuse, and the fusion rule can be written as
\begin{equation}
\alpha\times \beta = \sum_{\gamma} {N}_{\alpha\beta}^{\gamma} \, \gamma. \label{eq:boundary fusion rule}
\end{equation}
Here $\{N_{\alpha\beta}^{\gamma}\}$ is the set of fusion multiplicities of boundary topological excitations. They are non-negative integers satisfying the set of conditions in Appendix~\ref{Appendix_Fusion}. Note, however, that unlike the fusion multiplicities of anyons we have   ${N}_{\alpha\beta}^{\gamma}\ne {N}_{\beta \alpha}^{\gamma}$ for a most generic boundary theory. An example is the $K=\{1\}$ boundary of a quantum double model with a finite group $G$ \cite{2011CMaPh.306..663B,2019PhRvB..99c5112S}. In this case, each boundary superselection sector  can be identified as a group element of $G$, and the fusion rule is identical to the group multiplication. Therefore, the fusion is not commutative for a non-Abelian $G$.

The goal of this section is to provide a derivation of the $\ln d_{\alpha}$ universal topological contribution of von Neumann entropy given the fusion multiplicities $\{N_{\alpha\beta}^{\gamma} \}$. In particular, we will show that the $\ln d_{\alpha}$ contribution emerges as the only value consistent with SSA and a set of assumptions about the information convex. The method is parallel to the one discussed in Sec.~\ref{Sec_Bulk}.

\subsection{Preparing for the derivation}

In this subsection, we state and explain the key Assumption \ref{as:3'} which describes the way ${N}_{\alpha\beta}^{\gamma}$ is encoded in $\Sigma(\Omega_5)$. We further state proposition~\ref{as:4'} about the existence of a certain element which saturates SSA. The proof of proposition~\ref{as:4'} is provided in Sec.~\ref{Section_Info}. Both of them are crucial in the derivation of the $\ln d_{\alpha}$ contribution of the von Neumann entropy.

\begin{figure}[h]
	\centering
	\includegraphics[scale=1.0]{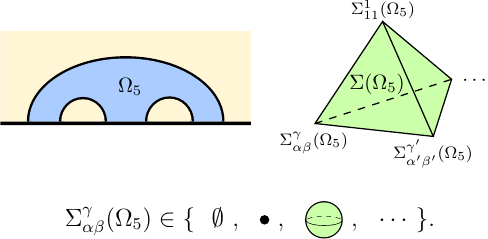}
	\caption{Subsystem $\Omega_5$ which touches a gapped boundary and its information convex $\Sigma(\Omega_5)$.	}\label{Boundary_Subsystems}
\end{figure}

\begin{figure}[h]
	\centering
	\includegraphics[scale=1.0]{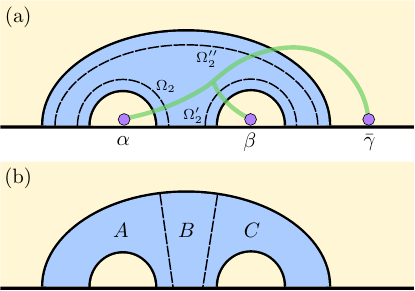}
	\caption{(a) Subsystems $\Omega_2, \Omega'_2, \Omega''_2 \subseteq \Omega_5$ surrounding the three entanglement cuts can detect the  boundary superselection sectors. In this picture, we have $\sigma_{\Omega_2}^{\alpha}$, $\sigma^{\beta}_{\Omega'_2}$, and $\sigma^{\gamma}_{\Omega''_2}$ on these subsystems. (b) $\Omega_5$ is divided into smaller pieces, $\Omega_5=ABC$. Both $AB$ and $BC$ have $\Omega_2$ topology. }
	\label{Boundary_ABC}
\end{figure}

\begin{assumption}{\Omega_5}{} \label{as:3'}
	For a subsystem $\Omega_5$ shown in Fig.~\ref{Boundary_Subsystems}, i.e., a connected subset of a half plane which has three entanglement cuts touching the same boundary, the information convex $\Sigma(\Omega_5)$ has the following structure:
	
	The set of extreme points of $\Sigma(\Omega_5)$ forms a set $\mathcal{M}=\bigcup_{(\alpha,\beta,\gamma)}\mathcal{M}^{\gamma}_{\alpha\beta}(\Omega_5)$. Each $\mathcal{M}^{\gamma}_{\alpha\beta}(\Omega_5)$ with ${N}_{\alpha\beta}^{\gamma} \ne 0$ is a connected component of $\mathcal{M}$. Let $\Sigma^{\gamma}_{\alpha\beta}(\Omega_5)$ be the convex subset of $\Sigma(\Omega_5)$  formed by the convex combination of elements in $\mathcal{M}^{\gamma}_{\alpha\beta}(\Omega_5)$. 
	
	Taking a partial trace to reduce $\sigma_{\Omega_5}^{(\alpha,\beta,\gamma)_x}\in \Sigma^{\gamma}_{\alpha\beta}(\Omega_5) $ to certain subsystems of topology $\Omega_2$ surrounding each entanglement cut, i.e., the $\Omega_2$, $\Omega'_2$ and $\Omega''_2$ shown in Fig.~\ref{Boundary_ABC}(a), we get extreme points  $\sigma^{\alpha}_{\Omega_2}$, $\sigma^{\beta}_{\Omega'_2}$ and $\sigma^{\gamma}_{\Omega''_2}$ respectively.
	
	Density matrix $\sigma_{\Omega_5}^{(\alpha,\beta,\gamma)_x}\in \Sigma^{\gamma}_{\alpha\beta}(\Omega_5)$ has one-to-one correspondence with $\hat{\rho}^{(\alpha,\beta,\gamma)_x}$, a  density matrix on a Hilbert space $\mathbb{V}_{\alpha\beta}^{\gamma} \equiv span \{ \vert i\rangle\}_{i=1}^{{N}_{\alpha\beta}^{\gamma}} $. Extreme points of $\Sigma^{\gamma}_{\alpha\beta}(\Omega_5)$ correspond to the pure-state density matrices on $\mathbb{V}_{\alpha\beta}^{\gamma} $.  The mapping preserves the convex structure, i.e.,
	\begin{equation}
	p\,\sigma_{\Omega_5}^{(\alpha,\beta,\gamma)_x} + (1-p) \sigma_{\Omega_5}^{(\alpha,\beta,\gamma)_y} =\sigma_{\Omega_5}^{(\alpha,\beta,\gamma)_z}  \nonumber
	\end{equation}
	if and only if 
	\begin{equation}
	p\,\hat{\rho}^{(\alpha,\beta,\gamma)_x} + (1-p) \hat{\rho}^{(\alpha,\beta,\gamma)_y} =\hat{\rho}^{(\alpha,\beta,\gamma)_z}, \nonumber
	\end{equation}
	where $p\in [0,1]$.
	Furthermore, $\sigma_{\Omega_5}^{(\alpha,\beta,\gamma)_x}$ has von Neumann entropy: 
	\begin{equation}
	\begin{aligned}
	S(\sigma_{\Omega_5}^{(\alpha,\beta,\gamma)_x}) =& S(\sigma_{\Omega_5}^{1}) + F(\alpha) + F(\beta) + F(\gamma)\\
	&+ S(\hat{\rho}^{(\alpha,\beta,\gamma)_x}).
	\end{aligned}
	\end{equation}
	Here $\sigma^1_{\Omega_5}\equiv \textrm{tr}_{\bar{\Omega}_5}\, \sigma^1$ is the unique element of $\Sigma^{1}_{11}(\Omega_5)$.
\end{assumption}

\begin{Proposition}  \label{as:4'}
	There exists a unique element $\sigma^{(\alpha,\beta)}_{\Omega_5} \in \Sigma(\Omega_5)$ such that for the partition $\Omega_5=ABC$ shown in Fig.~\ref{Boundary_ABC}(b) the following hold:
	\begin{enumerate}
		\item It prepares the extreme point $\sigma_{AB}^{\alpha}$ when restricted to $AB$
		and   it prepares the extreme point $\sigma_{BC}^{\beta}$ when restricted to $BC$. 
		\item It saturates the conditional mutual information:
		\begin{equation}
		I(A:C\vert B)_{\sigma^{(\alpha,\beta)}}=0.\nonumber
		\end{equation}
	\end{enumerate}
\end{Proposition}
The proof of proposition~\ref{as:4'} will be presented in Sec.~\ref{Section_Info}. For now, we simply point out that the proof needs \ref{as:G2}, \ref{as:S2}, \ref{as:2'} and SSA.

\subsection{ The derivation of the $\ln d_{\alpha}$  contribution to von Neumann entropy}

We present a derivation of the topological contribution to the von Neumann entropy on subsystem $\Omega_2$ by showing $F(\alpha)=\ln d_{\alpha}$, see Eq.~(\ref{eq:d_alpha}), and the probability in Eq.~(\ref{eq:boundary fusion prob}) from Assumptions \ref{as:G2}, \ref{as:S2}, \ref{as:1'}, \ref{as:2'} and \ref{as:3'}. The derivation is essentially the same with that in Sec.~\ref{Sec_Bulk_da}, and  we save our explanations in many places.

First, we apply Proposition~\ref{as:4'} to the $(\alpha,\beta)$ sector and the vacuum  sector $(1,1)$, to express the von Neumann entropy $S(\sigma^{(\alpha,\beta)}_{\Omega_5})$ and $S(\sigma^1_{\Omega_5})$ in terms of entropies on simpler subsystems $AB$, $BC$ and $B$ in Fig.~\ref{Boundary_ABC}(b).  We then use  \ref{as:2'} and \ref{as:1'} to compare the von Neumann entropy in these simpler subsystems and find
\begin{equation}
S(\sigma^{(\alpha,\beta)}_{\Omega_5})=S(\sigma^{1}_{\Omega_5}) + 2F(\alpha) + 2F(\beta). \label{star_2'}
\end{equation}
Note that $\sigma^{(1,1)}_{\Omega_5}= \sigma^1_{\Omega_5}$ follows from \ref{as:3'} and Proposition~\ref{as:4'}.

Next, we notice that  $\sigma^{(\alpha,\beta)}_{\Omega_5}$ defined in Proposition~\ref{as:4'} must be an element with maximal entropy among the elements of $\Sigma(\Omega_5)$ which have topological charge $\alpha$ and $\beta$ on the two entanglement cuts surrounded by $\Omega_2$ and $\Omega'_2$ in Fig.~\ref{Boundary_ABC}(a). This follows from SSA,  \ref{as:1'}, and \ref{as:2'}.
 With \ref{as:3'}, we express $\sigma^{(\alpha,\beta)}_{\Omega_5}$ as a convex combination of elements in $\bigcup_{\gamma}\Sigma_{\alpha\beta}^{\gamma}(\Omega_5)$, and calculate its entropy with the help of \ref{as:2'}. By maximizing the von Neumann entropy, we find
\begin{equation}
S(\sigma_{\Omega_5}^{(\alpha,\beta)})= S(\sigma^1_{\Omega_5}) +F(\alpha) + F(\beta) + \ln (\sum_{\gamma} N_{\alpha\beta}^{\gamma}\, e^{F(\gamma)}) \label{star_1'}
\end{equation}
and 
\begin{equation}
\sigma^{(\alpha,\beta)}_{\Omega_5}= \sum_{\gamma} P_{(\alpha\times\beta\to\gamma)} \sigma_{\Omega_5}^{{(\alpha,\beta,\gamma)}_{\textrm{max}} }, \label{eq:sigma(ab)'}
\end{equation}
where $\sigma_{\Omega_5}^{(\alpha,\beta,\gamma)_{\textrm{max}}}$ is the maximal-entropy element of $\Sigma^{\gamma}_{\alpha\beta}(\Omega_5)$ and
$P_{(\alpha\times\beta\to\gamma)}=\frac{ N_{\alpha\beta}^{\gamma}e^{F(\gamma)} }{ \sum_{\delta} N_{\alpha\beta}^{\delta}e^{F(\delta)} }$.

We have obtained two expressions of $S(\sigma_{\Omega_5}^{(\alpha,\beta)})$ in Eq.~(\ref{star_2'}) and  Eq.~(\ref{star_1'}). By comparing them, one finds 
\begin{equation}
e^{F(\alpha)} e^{F(\beta)}= \sum_{\gamma} N_{\alpha\beta}^{\gamma} e^{F(\gamma)}. \label{eq:e^F(alpha)}
\end{equation}
Because $F(\alpha)$ is real, we must have $e^{F(a)}\in (0,+\infty)$. 
Equation~(\ref{eq:e^F(alpha)}) has a unique solution $e^{F(\alpha)}=d_{\alpha}$, where $d_{\alpha}$ is the quantum dimension uniquely defined given $N_{\alpha\beta}^{\gamma}$. See Appendix \ref{Appendix_Fusion} for a self-contained proof using a few assumptions about the fusion multiplicities. Thus,
\begin{equation}
F(\alpha)=\ln d_{\alpha}. \label{eq:d_alpha}
\end{equation}

To summarize, we have derived the $\ln d_{\alpha}$ topological contribution of von Neumann entropy as the only value consistent with SSA given Assumptions \ref{as:G2}, \ref{as:S2}, \ref{as:1'}, \ref{as:2'} and \ref{as:3'}. We have also derived the explicit form of density matrix  $\sigma^{(\alpha,\beta)}_{\Omega_5}$ in Eq.~(\ref{eq:sigma(ab)'}) with probability
\begin{equation}
P_{(\alpha\times\beta\to\gamma)}=\frac{N_{\alpha \beta}^{\gamma} d_{\gamma}}{d_{\alpha} d_{\beta}}. \label{eq:boundary fusion prob}
\end{equation}
Later  we will identify the physical meaning of probability $P_{(\alpha\times \beta\to \gamma)}$; see Sec.~\ref{Sec_Fusion probability}.

\section{Condensation from the bulk to a gapped boundary}\label{Section:condensation}
We know from the previous discussion that for a nonchiral topological order with a gapped boundary there is a set of bulk superselection sectors $\{1,a,b,\cdots \}$ and a set of boundary superselection sectors $\{ 1,\alpha,\beta,\cdots \}$. It is possible to move an anyon $a$ onto a gapped boundary. After measurement of the boundary superselection sector, it will turn into a certain boundary topological excitation type. 
One could formally write down this process as
\begin{equation}
a=\sum_{\alpha} N^{\alpha}_{a} \alpha.\label{eq:condensation_rule}
\end{equation}
This process is similar to fusion [Eq. (\ref{eq:bulk fusion rule}) and Eq.~(\ref{eq:boundary fusion rule})] but there is only one excitation on the left-hand side. We will call this process \emph{bulk-to-boundary condensation} \footnote{The terminology is adapted from Ref.~\cite{2017CMaPh.355..645C}, the algebraic result therein is the same with ours. However, the physical context is different since the boundary excitations in \cite{2017CMaPh.355..645C} are confined.} and we will frequently call it \emph{condensation} for short. We call Eq.~(\ref{eq:condensation_rule}) the condensation rule. Here $\{N^{\alpha}_{a}\}$ is a set non-negative integers which we call the condensation multiplicities. We say boundary topological excitation $\alpha$ is in the condensation channel of anyon $a$ if $N^{\alpha}_{a}\ge 1$.

Because the terminology \emph{bulk-to-boundary condensation} is not as standard as \emph{fusion}, we provide some additional explanations and point out alternative terminologies in the literature.  In Ref.~\cite{2011CMaPh.306..663B}, an anyon is said to condense onto a gapped boundary if it can disappear on a gapped boundary \footnote{It is called ``annihilating" particles at a gapped boundary in Ref.~\cite{Levin2013}.}. This corresponds to the condensation to the boundary vacuum  $1$ in our terminology. Therefore, the condensation rule in our definition is a generalization which includes a generic boundary superselection sector. Our bulk-to-boundary condensation has the same meaning as the \emph{bulk-to-boundary map} or ``the bulk excitations fuse into the boundary" in \cite{2012CMaPh.313..351K,2014NuPhB.886..436K}. Also, it is worth noting that bulk-to-boundary condensation is physically different from anyon condensation \cite{2009PhRvB..79d5316B,2009PhRvL.102v0403B,2014NuPhB.886..436K,2014PhRvB..90s5130E}, which is a relation between two topologically ordered phases of the bulk.

The goal of this section is to provide a derivation of the formula $d_{a}=\sum_{\alpha} N_{a}^{\alpha} d_{\alpha}$. It emerges from the consistency with SSA and a set of assumptions about the information convex. The method is parallel to the one discussed in Sec.~\ref{Sec_Bulk}.

\begin{figure}[h]
	\centering
	\includegraphics[scale=1.0]{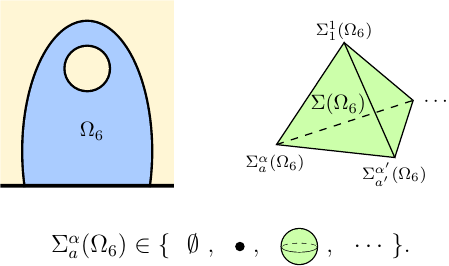}
	\caption{Subsystem $\Omega_6$ which touches a gapped boundary and its information convex $\Sigma(\Omega_6)$.}	\label{Bulk_Boundary_Subsystems}
\end{figure}

\subsection{Preparing for the derivation}

In this subsection, we state and explain the key assumption \ref{as:1''} which describes the way ${N}_{a}^{\alpha}$ is encoded in $\Sigma(\Omega_6)$. We further state Proposition~\ref{as:2''} about the existence of a certain element which saturates SSA. The proof of Proposition~\ref{as:2''} is provided in Sec.~\ref{Section_Info}. Both of them are crucial in the derivation of the formula $d_{a}=\sum_{\alpha} N_{a}^{\alpha} d_{\alpha}$.

\begin{assumption}{\Omega_6}{}  \label{as:1''}
	Let $\Omega_6$ be a connected subsystem with two entanglement cuts, one in the bulk and another touching the gapped boundary; see Fig.~\ref{Bulk_Boundary_Subsystems}. Its information convex $\Sigma(\Omega_6)$ has the following structure:
	
	The set of extreme points of $\Sigma(\Omega_6)$ forms a set $\mathcal{M}=\bigcup_{(a,\alpha)}\mathcal{M}_{a}^{\alpha}(\Omega_6)$. Each $\mathcal{M}_{a}^{\alpha}(\Omega_6)$ with $N_{a}^{\alpha} \ne 0$ is a connected component of $\mathcal{M}$. Let $\Sigma^{\alpha}_{a}(\Omega_6)$ be the convex subset of $\Sigma(\Omega_6)$  formed by a convex combination of elements in $\mathcal{M}_{a}^{\alpha}(\Omega_6)$. 
	
	Taking a partial trace to reduce  $\sigma_{\Omega_6}^{(a,\alpha)_x}\in \Sigma^{\alpha}_{a}(\Omega_6) $ and  $\Omega_1$ or $\Omega_2$ surrounding the entanglement cuts, see Fig.~\ref{Bulk_Boundary_ABC}(a), we get the extreme point $\sigma^{a}_{\Omega_1} \in \Sigma(\Omega_1)$ and the extreme point $\sigma^{\alpha}_{\Omega_2} \in \Sigma(\Omega_2)$, respectively.

	Density matrix $\sigma_{\Omega_6}^{(a,\alpha)_x}\in \Sigma^{\alpha}_{a}(\Omega_6)$ has a one-to-one correspondence with $\hat{\rho}^{(a,\alpha)_x}$, a  density matrix on a Hilbert space $\mathbb{V}_{a}^{\alpha} \equiv span \{ \vert i\rangle\}_{i=1}^{N_{a}^{\alpha}} $. Extreme points of $\Sigma_{a}^{\alpha}(\Omega_6)$ correspond to the pure-state density matrices on $\mathbb{V}_{a}^{\alpha}$.  The mapping preserves the convex structure, i.e.,
	\begin{equation}
	p\,\sigma_{\Omega_6}^{(a,\alpha)_x} + (1-p) \sigma_{\Omega_6}^{(a,\alpha)_y} =\sigma_{\Omega_6}^{(a,\alpha)_z} \nonumber
	\end{equation}
	if and only if 
	\begin{equation}
	p\,\hat{\rho}^{(a,\alpha)_x} + (1-p) \hat{\rho}^{(a,\alpha)_y} =\hat{\rho}^{(a,\alpha)_z}, \nonumber
	\end{equation}
	where $p\in [0,1]$.
	Furthermore, $\sigma_{\Omega_6}^{(a,\alpha)_x}$ has von Neumann entropy: 
	\begin{equation}
	S(\sigma_{\Omega_6}^{(a,\alpha)_x}) = S(\sigma_{\Omega_6}^{1}) + f(a) + F(\alpha) 
	+ S(\hat{\rho}^{(a,\alpha)_x}).
	\end{equation}
	Here $\sigma^1_{\Omega_6}\equiv \textrm{tr}_{\bar{\Omega}_6}\,\sigma^1$ is the unique element of $\Sigma_{1}^{1}(\Omega_6)$. 
\end{assumption}

\begin{figure}[h]
	\centering
	\includegraphics[scale=1.0]{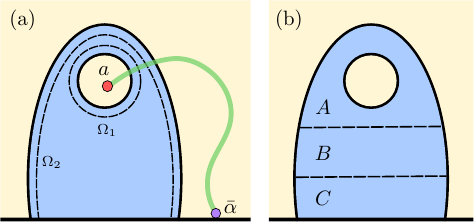}
	\caption{(a) Subsystems $\Omega_1,\Omega_2\subseteq \Omega_6$ surrounding the two entanglement cuts. For the state shown in the diagram, the topological charges can be read from the extreme points $\sigma^a_{\Omega_1}$ and $\sigma^{\alpha}_{\Omega_2}$. (b) $\Omega_6$ is divided into smaller pieces $\Omega_6=ABC$. $AB$ has $\Omega_1$ topology and $BC$ has $\omega_2$ topology.}
	\label{Bulk_Boundary_ABC}
\end{figure}

\begin{Proposition}  \label{as:2''}
	There exists a unique element $\sigma^{(a)}_{\Omega_6} \in \Sigma(\Omega_6)$  such that for the partition $\Omega_6=ABC$ shown in Fig.~\ref{Bulk_Boundary_ABC}(b) the following hold:
	\begin{enumerate}
		\item It prepares the extreme point $\sigma_{AB}^{a}$ when restricted to $AB$.

		\item It saturates the conditional mutual information:
		\begin{equation}
		I(A:C\vert B)_{\sigma^{(a)}}=0.\nonumber
		\end{equation}
	\end{enumerate}
\end{Proposition}
The proof of proposition~\ref{as:2''} will be presented in Sec.~\ref{Section_Info}. For now, we simply point out that the proof needs \ref{as:G3}, \ref{as:S1}, \ref{as:2} and SSA.

\subsection{The derivation of $d_a=\sum_{\alpha} N_a^{\alpha} d_{\alpha}$}

We present a derivation of the formula $d_a=\sum_{\alpha} N_a^{\alpha} d_{\alpha}$, see Eq.~(\ref{eq:conservation of quantum dimension}), and the probability in Eq.~(\ref{Condensation_Prob}) from Assumptions \ref{as:G3}, \ref{as:S1}, \ref{as:1}, \ref{as:1'}, \ref{as:2}, \ref{as:2'} and \ref{as:1''}. Again, the derivation is essentially the same as that in Sec.~\ref{Sec_Bulk_da}, and we save our explanations in many places.

First, we apply Proposition~\ref{as:2''} to both $(a,\alpha)$ and $(1,1)$ sectors and use \ref{as:1}, \ref{as:1'} and \ref{as:2}  to derive
\begin{equation}
S(\sigma^{(a)}_{\Omega_6}) =S(\sigma^1_{\Omega_6}) +2 f(a). \label{eq:condensation-2}
\end{equation}
Note that $\sigma^{(1)}_{\Omega_6}=\sigma^1_{\Omega_6}$ follows from \ref{as:1''} and proposition~\ref{as:2''}.

Next, we notice that  $\sigma^{(a)}_{\Omega_6}$ defined in proposition~\ref{as:2''} must be an element with maximal entropy among the elements of $\Sigma(\Omega_6)$ which have topological charge $a$  on the  entanglement cut surrounded by $\Omega_1$ in Fig.~\ref{Bulk_Boundary_ABC}(a). This follows from SSA,  \ref{as:1}, \ref{as:1'} and \ref{as:2}.
With \ref{as:1''}, we express $\sigma^{(a)}_{\Omega_6}$ as a convex combination of elements in $\bigcup_{\alpha} \Sigma_{a}^{\alpha}(\Omega_6)$, and calculate its entropy with the help of \ref{as:2} and \ref{as:2'}. By maximizing the von Neumann entropy, we find

\begin{equation}
S(\sigma^{(a)}_{\Omega_6} ) =S(\sigma^1_{\Omega_6}) + f(a) +\ln (\sum_{\alpha} N_{a}^{\alpha} e^{F(\alpha)} ) \label{eq:condensation-1}
\end{equation}
and
\begin{equation}
\sigma^{(a)}_{\Omega_6} =\sum_{\alpha} P_{(a\to\alpha)} \sigma_{\Omega_6}^{(a,\alpha)_{\textrm{max}}}, \label{eq:p6}
\end{equation} 
where  $\sigma_{\Omega_6}^{(a,\alpha)_{\textrm{max}}}$ is the maximal-entropy element of $\Sigma^{\alpha}_{a}(\Omega_6)$ and
$P_{(a \to\alpha)} =\frac{N_{a}^{\alpha} e^{F(\alpha)}}{ \sum_{\beta} N_{a}^{\beta} e^{F(\beta)} }$.

We have obtained two expressions of $S(\sigma^{(a)}_{\Omega_6} )$ in Eq.~(\ref{eq:condensation-2}) and Eq.~(\ref{eq:condensation-1}). By comparing them and plugging in the values $f(a)=\ln d_a$ and $F(\alpha)=\ln d_{\alpha}$ obtained in Eq. (\ref{eq:d_a}) and Eq. (\ref{eq:d_alpha}), we get 
\begin{equation}
d_{a} =\sum_{\alpha} N_{a}^{\alpha} d_{\alpha}. \label{eq:conservation of quantum dimension}
\end{equation}
Now we can rewrite the probability in Eq.~(\ref{eq:p6}) as 
\begin{equation}
P_{(a\to\alpha)} = \frac{N_a^{\alpha} d_{\alpha} }{d_a}. \label{Condensation_Prob}
\end{equation}
Later we will identify the physical meaning of probability $P_{(a\to\alpha)}$; see Sec.~\ref{Sec_Fusion probability}.

\section{Some proofs, probabilities, circuit depth and more}\label{Section_Info}

In this section, we present a collection of results about the quantum states and the universal properties of topologically ordered systems. Note that many of the results have been obtained from other methods. The primary purpose is to present a different logic, namely, that these results can be derived from information-theoretic constraints.
After reviewing some useful facts about conditional mutual information in Sec.~\ref{Sec_info_1}, we provide proofs of Proposition~\ref{as:4}, \ref{as:4'}, and \ref{as:2''} in Sec.~\ref{Sec_info_2}. In Sec.~\ref{Sec_Fusion probability}, we identify the physical meanings of the probabilities $P_{(a\times b\to c)}$, $P_{(\alpha\times\beta\to\gamma)}$ and $P_{(a\to\alpha)}$ [from Eq.~(\ref{eq:Bulk Fusion Prob}), (\ref{eq:boundary fusion prob}), and (\ref{Condensation_Prob}), respectively]. In Sec.~\ref{Sec._Depth}, we discuss the circuit depth of unitary string operators. In Sec.~\ref{Sec. Markov}, we discuss some additional results which rely on a more advanced result, namely, the ability to merge certain quantum Markov states.

\subsection{Some useful facts about conditional mutual information}\label{Sec_info_1}

Recall that SSA \cite{1973JMP....14.1938L} says
$I(A:C\vert B)_{\rho}\ge 0$, $\forall\,\rho_{ABC}$. A state $\rho_{ABC}$ saturates SSA if $I(A:C\vert B)_{\rho}=0$. Such a state  is called a quantum Markov state (QMS)\cite{Petz1987,2004CMaPh.246..359H}.

Some basic properties for a general state are as follows:
\begin{enumerate}
	\item Let $A=A_1A_2$ and $C=C_1C_2$; then
	\begin{eqnarray}
     I(A_2:C_2\vert B) &\le&	I(A:C\vert B)   \label{Ineq_1}\\
	 I(A_1:C_1\vert A_2 B C_2) &\le&  I(A:C\vert B) \label{Ineq_2}.
	\end{eqnarray}
	\item If $\sigma$ and $\sigma'$ are related by unitary transformations on individual subsystems, i.e.,
	\begin{equation}
	\sigma'_{ABC}=(U_A U_B U_C) \,\sigma_{ABC} \, (U_{A} U_B U_C)^{\dagger}\nonumber
	\end{equation}
	 then
	\begin{equation}
	I(A:C\vert B)_{\sigma}=I(A:C\vert B)_{\sigma'}. \label{Equality_Unitary}
	\end{equation}
\end{enumerate}
Inequalities (\ref{Ineq_1}) and (\ref{Ineq_2}) are derived using SSA.
For a QMS $\rho$, $I(A:C\vert B)_{\rho}=0$ and the inequalities above become equalities, i.e.,
\begin{eqnarray}
 I(A_2:C_2\vert B)_{\rho} &=& I(A:C\vert B)_{\rho} =0, \label{Equality_trace}\\
 I(A_1:C_1\vert A_2 B C_2)_{\rho} &=&  I(A:C\vert B)_{\rho} =0. \label{Equality_move}
\end{eqnarray}
We further notice that a QMS $\rho_{ABC}$ with $I(A:C\vert B)_{\rho}=0$ is uniquely determined by its reduced density matrices $\rho_{AB}$ and $\rho_{BC}$ \footnote{There is a explicit expression of the state $\rho_{ABC}=\rho_{AB}^{\frac{1}{2}}\rho_{B}^{-\frac{1}{2}}\rho_{BC}\rho_{B}^{-\frac{1}{2}}\rho_{AB}^{\frac{1}{2}}$ \cite{2003RvMaP..15...79P}. We also notice a robust version of this statement proved using Kim-Ruskai inequality \cite{2014JMP....55i2201K}, see Theorem 1 of \cite{2014arXiv1405.0137K}.  }.

\subsection{The proof of \ref{as:4}, \ref{as:4'} and \ref{as:2''} }\label{Sec_info_2}

In this subsection, we provide the proofs of propositions \ref{as:4}, \ref{as:4'} and \ref{as:2''} using assumptions described in Sec.~\ref{ground state and excitations}. More precisely, we will show the following: 
\begin{enumerate}
	\item \{ \ref{as:G1}, \ref{as:S1}, \ref{as:2} \} $\Rightarrow$ proposition \ref{as:4}.
	\item \{ \ref{as:G2}, \ref{as:S2}, \ref{as:2'} \} $\Rightarrow$ proposition \ref{as:4'}.
	\item \{ \ref{as:G3}, \ref{as:S1}, \ref{as:2} \} $\Rightarrow$ proposition \ref{as:2''}.
\end{enumerate}
 Because the three proofs are essentially the same, we choose to discuss the proof of Proposition~\ref{as:4} in detail and then briefly discuss the proof of the other two propositions.

\begin{figure}[h]
	\centering
	\includegraphics[scale=1.0]{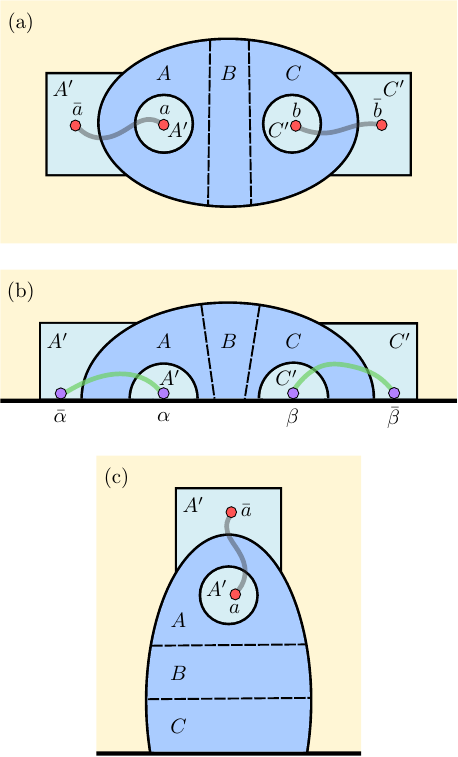}
	\caption{(a) Subsystems $AA',B,CC'$ are in the bulk. The unitary string operator $U^{(a,\bar{a})}$ lies in $AA'$ and the anyon pair $(a,\bar{a})$ it creates is located within $A'$. The unitary string operator $U^{(b,\bar{b})}$ lies in $CC'$ and the anyon pair $(b,\bar{b})$ it creates is located within $C'$. (b) Subsystems $AA',B,CC'$ are attached to a gapped boundary.  The unitary string operator $U^{(\alpha,\bar{\alpha})}$ lies in $AA'$ and the boundary topological excitations  $(\alpha,\bar{\alpha})$ it creates are  located within $A'$. The unitary string operator $U^{(\beta,\bar{\beta})}$ lies in $CC'$ and the boundary topological excitations $(\beta,\bar{\beta})$ it creates are located within $C'$.  (c) Subsystems $AA',B,C$ where $C$ attaches to the boundary and $AA'B$ is in the bulk (away from the boundary).  The unitary string operator $U^{(a,\bar{a})}$ lies in $AA'$ and the anyons   $(a,\bar{a})$ it creates are located within $A'$. }
	\label{44_Proof}
\end{figure}

 The proof of Proposition~\ref{as:4} contains an explicit construction of the state $\sigma^{(a,b)}_{\Omega_4}$. Let us consider the bulk subsystems $AA'$, $B$, $CC'$ shown in Fig.~\ref{44_Proof}(a), $\Omega_4=ABC$. According to \ref{as:G1},
\begin{equation}
 I(AA' : CC'\vert B)_{\sigma^1}=0.\nonumber
\end{equation}
Let us apply unitary string operators $U^{(a,\bar{a})}$ and $U^{(b,\bar{b})}$, the support of which is within $AA'$ and $CC'$, respectively. These unitary string operators are guaranteed to exist according to \ref{as:S1}. The anyon pair $(a,\bar{a})$  is in $A'$ and the anyon pair $(b,\bar{b})$  is in $C'$. Let
\begin{equation}
{\sigma}^{(a,b)}_{\Omega_4}\equiv\textrm{tr}_{A'C'}\big[ (U^{(a,\bar{a})} U^{(b,\bar{b})}) \sigma^1_{AA' BCC'}  (U^{(a,\bar{a})} U^{(b,\bar{b})})^{\dagger} \big]. \nonumber
\end{equation}
According to Eq. (\ref{Equality_Unitary}) and Eq. (\ref{Equality_trace}), we have
\begin{equation}
I(A:C\vert B)_{{\sigma}^{(a,b)}}=0.\nonumber
\end{equation}
Moreover, ${\sigma}^{(a,b)}_{\Omega_4}\in \Sigma(\Omega_4)$ since there are no excitations within and around $ABC$ before tracing out $A'$ and $C'$.
 Also notice that ${\sigma}^{(a,b)}_{\Omega_4}$ has topological charges $a$, $b$ in the two holes and therefore, according to \ref{as:2}, it prepares the extreme points $\sigma^a_{AB}$ and $\sigma^b_{BC}$. 
 Thus, the density matrix ${\sigma}^{(a,b)}_{\Omega_4}$ satisfies all the conditions required in proposition~\ref{as:4}. Furthermore, this QMS state ${\sigma}^{(a,b)}_{\Omega_4}$ is uniquely determined by its reduced density matrices  $\sigma^a_{AB}$ and $\sigma^b_{BC}$. This completes the proof of proposition~\ref{as:4}.

This proof provides a physical construction of $\sigma^{(a,b)}_{\Omega_4}$. After comparing it with the construction using the extreme points of $\Sigma(\Omega_4)$, in Eq. (\ref{eq:sigma(ab)}),
 one can identify the physical interpretation of $P_{(a\times b\to c)}$; see Sec.~\ref{Sec_Fusion probability}.

The proof of proposition~\ref{as:4'} is completely parallel given \ref{as:G2}, \ref{as:S2}, \ref{as:2'} and the idea is illustrated in Fig.~\ref{44_Proof}(b). Furthermore, the proof of proposition~\ref{as:2''} is also parallel given \ref{as:G3}, \ref{as:S1}, \ref{as:2} and the idea is illustrated in Fig.~\ref{44_Proof}(c). 

In comparison, we notice that applying $U^{(a,\bar{\alpha})}$ onto the ground state, (which is guaranteed to exist by \ref{as:S3}), as is shown in Fig.~\ref{Bulk_Boundary_ABC}(a), in general does not give us a state with vanishing conditional mutual information for the partition in Fig.~\ref{Bulk_Boundary_ABC}(b).

\subsection{Physical interpretation of probabilities}\label{Sec_Fusion probability}

In this subsection, we discuss the physical interpretations for the probabilities $P_{(a\times b\to c)}$ in Eq. (\ref{eq:Bulk Fusion Prob}),  $P_{(\alpha\times\beta\to\gamma)}$ in Eq. (\ref{eq:boundary fusion prob}), and $P_{(a\to\alpha)}$ in Eq. (\ref{Condensation_Prob}).

The physical interpretation for the probability
\begin{equation}
P_{(a\times b\to c)}= \frac{  N_{ab}^c d_c }{ d_a d_b }\nonumber
\end{equation}
is the probability for a pair of independently created anyons $a$, $b$ to fuse into $c$. Here, we say two anyons are independently created if they are created separately by two string operators acting on nonoverlapping supports; see Fig.~\ref{44_Proof}(a). We will call $P_{(a\times b\to c)}$ the \emph{fusion probability}.
 This physical interpretation has been explored in the literature; see \cite{preskill1998lecture}.
The information-theoretic considerations in this paper rederive this physical interpretation. In more detail, this physical interpretation follows from two facts:
\begin{enumerate}
	\item $\sigma^{(a,b)}_{\Omega_4}$ is the reduced density matrix of a state with $(a,\bar{a})$, $(b,\bar{b})$ created by two separated unitary strings on the ground state; see Fig.~\ref{44_Proof}(a).
	\item By restricting $\sigma^{(a,b)}_{\Omega_4}$ onto $\Omega''_1$, see Fig.~\ref{Bulk_Restriction}(a), we get
	\begin{equation}
	\textrm{tr}_{\Omega_4 \backslash \Omega''_1}\, \sigma^{(a,b)}_{\Omega_4} = \sum_{c} P_{(a\times b\to c)} \sigma^c_{\Omega''_1}. \nonumber
	\end{equation}
\end{enumerate}

From the same reasoning, one arrives at the physical interpretation for the probability
\begin{equation}
P_{(\alpha\times\beta\to\gamma)}=\frac{  N_{\alpha\beta}^{\gamma} d_{\gamma} }{d_{\alpha} d_{\beta}}\nonumber
\end{equation}
as the probability for a pair of independently created boundary excitations $\alpha$, $\beta$ to fuse into $\gamma$. Here, we say a pair of boundary excitations are independently created if they are created separately by two string operators acting on nonoverlapping supports; see Fig.~\ref{44_Proof}(b). We will call $P_{(\alpha\times\beta\to\gamma)}$ the \emph{fusion probability}.

Furthermore, it can be shown by the same method that the physical interpretation for the probability
\begin{equation}
	P_{(a\to \alpha)}= \frac{N_a^{\alpha}{ d_{\alpha} }}{ d_a}\nonumber
\end{equation}
 is  the probability for a bulk anyon $a$, which had never met the boundary, to condense into a boundary topological excitation $\alpha$.
 Here, we say that an anyon had never met the boundary if it is created by a string operator supported in the bulk (away from the boundary); see Fig.~\ref{44_Proof}(c). We call $P_{(a\to \alpha)}$ the \emph{condensation probability}.
  In contrast, the anyon $a$ in  Fig.~\ref{Bulk_Boundary_ABC}(a) has met the boundary.

\subsection{Circuit depth of string operators} \label{Sec._Depth}

The circuit depth of a unitary operator is a measure of how complex a unitary operator is, from the viewpoint of a fixed basis (i.e., the real space). In this section, we consider how complex a unitary operator we need in order to propagate topological excitations by a certain distance. Therefore, we discuss the  circuit depth of the unitary string operators which create these excitations, both in the bulk and in the presence of a gapped boundary. In particular, we will consider the unitary string operators $U^{(a,\bar{a})}$  and  $U^{(\alpha,\bar{\alpha})}$ which  are  guaranteed to exist by \ref{as:S1}, \ref{as:S2}.  We further consider the generalizations shown in Fig.~\ref{Unitary_String_Depth}.

The discussion in this section is of the same spirit as \cite{2010PhRvB..82o5138C,Haah2016} which consider the minimal circuit depth needed in order to convert a topologically ordered ground state into a product state. Because two states within the same gapped phase may (approximately) be related by a finite-depth quantum circuit, the circuit depth of unitary string operators (either finite depth  or a depth scale with some length scale) is a universal property of a gapped phase.

For an Abelian anyon $a$,  $U^{(a,\bar{a})}$ is consistent with a finite-depth quantum circuit. In fact, in many exactly solvable models, an Abelian anyon string is a depth-1 quantum circuit. The depth is independent of the string length, i.e., the distance separation between the anyons.
However, we do not know whether Abelian strings are finite-depth quantum circuits in general.

In this subsection, we derive a concrete result: the circuit depth of a non-Abelian anyon string is at least linear to the distance separation of the anyon pair \footnote{This linear bound might be related to a property of Wilson loop operator in non-Abelian gauge theory  \cite{2002PhRvD..65f5022B}.}. 
We further discuss how the result generalizes and changes in the presence of a gapped boundary.

\begin{figure}[h]
	\centering
	\includegraphics[scale=1.0]{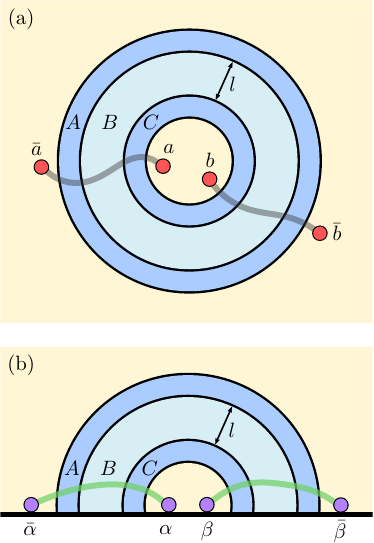}
	\caption{An illustration of the subsystem choices, unitary string operators and the anyons or boundary topological excitations the string operators create. $l>\epsilon$ is the lattice distance between $A$ and $C$. (a) $ABC$ and $A$, $B$, $C$ are subsystems of $\Omega_1$ topology. They are in the bulk. (b) $ABC$ and $A$, $B$, $C$ are subsystems of $\Omega_2$ topology. They attach to a gapped boundary.}
	\label{Probability_1}	
\end{figure}

Let us consider the state with bulk string $U^{(a,\bar{a})}$  and $U^{(b,\bar{b})}$ acting on a ground state shown in Fig.~\ref{Probability_1}(a). The subsystems $A$, $C$, and $ABC$ are of the same topology, i.e., the topology of $\Omega_1$. According to the discussion in Sec.~\ref{Sec_Fusion probability}, the density matrices on $A$, $C$ and $ABC$ are the  following convex combination of extreme points:
\begin{equation}
\sigma_{\Omega_1}^{(a,b)}= \sum_c P_{(a\times b\to c)}\sigma^c_{\Omega_1},\quad \Omega_1 = A,C,ABC,
\end{equation}
where $P_{(a\times b\to c)}=\frac{N_{ab}^c d_c}{d_a d_b}$. In particular, let us consider $b=\bar{a}$ and $a$ is non-Abelian. The fusion rule says
\begin{equation}
a\times \bar{a} =1 + \cdots.
\end{equation}
For any non-Abelian anyon $a$, we always have $d_a>1$ and $P_{(a \times \bar{a}\to 1)}=1/{d_a^2}<1$ and there must be fusion results other than the vacuum $1$. 
Thus,
\begin{equation}
\begin{aligned}
&I(A:C)_{\sigma^{(a,\bar{a})}_{ABC}}\\
&= \sum_c P_{(a\times\bar{a}\to c)} (I(A:C)_{\sigma^c} -\ln P_{(a\times\bar{a}\to c)})\\
&\ge -\sum_c P_{(a\times \bar{a}\to c)}\ln P_{(a\times \bar{a}\to c)}\\
&> 0. \label{Mutual_Info_Nontrivial}
\end{aligned}
\end{equation}
In the second line, we used \ref{as:2}. In the third line, we used $I(A:C)\ge 0$ for any state \footnote{In fact, it could be replaced by ``$=$" since it is known that  $I(A:C)_{\sigma^{c}_{ABC}}=0$ for each extreme point (from topological field theory and replica trick \cite{2015arXiv150807006J}, and it is also possible to verify this result in lattice models). Nevertheless, we do not need this property in this proof.}.  In the fourth line, we used the fact that $a$ is non-Abelian.
On the other hand,  \ref{as:G0} tells us
\begin{equation}
I(A:C)_{\sigma^1}=0. \label{Mutual_Info_trivial}
\end{equation}
Comparing Eq. (\ref{Mutual_Info_trivial}) with Eq. (\ref{Mutual_Info_Nontrivial}), we find that the mutual information $I(A:C)$ is changed from zero to a positive number. This cannot be done by any quantum circuit with depth smaller than $(l-\epsilon)/2$. Here $l$ is the minimal lattice distance between $A$ and $C$ shown in Fig.~\ref{Probability_1}(a).  For a large separation between $a$ and $\bar{a}$, $l$ can be approximately as large as the separation.

To summarize, we have proved that the unitary operator $U^{(a,\bar{a})}$ has a circuit depth at least linear in the distance between $a$ and $\bar{a}$ if the anyon  $a$ is non-Abelian.

In the presence of a gapped boundary, the above result has a generalization. Meanwhile, one should be aware of a certain difference.

First, for the case shown in Fig.~\ref{Probability_1}(b), the unitary string operator $U^{(\alpha,\bar{\alpha})}$ has a circuit depth at least linear in the distance between $\alpha$ and $\bar{\alpha}$ if the boundary topological excitation $\alpha$ is non-Abelian.  This fact follows from the same logic as above.

Second, in the presence of a gapped boundary, e.g., the configuration in Fig.~\ref{Unitary_String_Depth}(b), the linear bound for the circuit depth of a string operator [creating a non-Abelian $(a,\bar{a})$ pair] may be relaxed. This happens when the following two conditions are satisfied:
\begin{enumerate}
	\item There is an Abelian boundary topological excitation $\alpha$ in the condensation channel of $a$.
	\item The Abelian boundary topological excitation pair $(\alpha,\bar{\alpha})$ can be created with a constant-depth quantum circuit.
\end{enumerate}
Under these two conditions, the minimal circuit depth of  $U^{(a,\bar{\alpha})}$ and $\tilde{U}^{(a,\bar{a})}$ in Fig.~\ref{Unitary_String_Depth} would scale with $l$ instead of $L$.

\begin{figure}[h]
	\centering
	\includegraphics[scale=1.0]{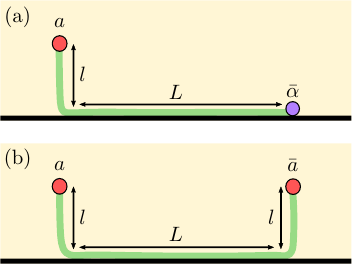}
	\caption{String operators touching a gapped boundary. An Abelian boundary topological excitation $\alpha$ is in the condensation channel of a non-Abelian anyon $a$. (a) The string  operator $U^{(a,\bar{\alpha})}$. (b) The string operator $\tilde{U}^{(a,\bar{a})}$.}
	\label{Unitary_String_Depth}	
\end{figure}

\subsection{Additional results based on the merging of quantum Markov states}\label{Sec. Markov}

In this subsection, we derive some additional results using a deeper property of QMSs. The key technique comes from \cite{2016PhRvA..93b2317K}, in which  the authors show  how to merge a pair of QMSs into a global state. 
Recall that a density matrix $\rho_{ABC}$ is a QMS if $I(A:C\vert B)_{\rho}=0$.

For the subsystems $A$, $B_1B_2$, $C$ in Fig.~\ref{B1B2}, the ground-state density matrix $\sigma^1$ satisfies
	\begin{eqnarray}
	I(A:B_2\vert B_1)_{\sigma^1}&=& I(B_1:C\vert B_2)_{\sigma^1}=0\nonumber\\
	\sigma^1_{AB_2}&=&\sigma^1_{A}\otimes \sigma^1_{B_2}.\nonumber
	\end{eqnarray}
Then it follows from the construction in \cite{2016PhRvA..93b2317K} that
    there exists a state $\tilde{\sigma}_{ABC}$ ($B=B_1B_2$) such that
	\begin{eqnarray}
	\tilde{\sigma}_{AB} &=& \sigma^1_{AB}, \nonumber\\
	\tilde{\sigma}_{BC} &=& \sigma^1_{BC},\nonumber \\
	I(A:C\vert B)_{\tilde{\sigma}}&=&0, \nonumber \\
	\tilde{\sigma}_{AC} &=& \sigma^1_{A}\otimes \sigma^1_{C}= \sigma^1_{AC}. \nonumber
	\end{eqnarray}
	Note that the  discussion applies to all three cases Fig.~\ref{B1B2}(a), Fig.~\ref{B1B2}(b), and Fig.~\ref{B1B2}(c).

For the  subsystems $A$, $B$, $C$ in Fig.~\ref{2nd_way_to_divide_Omega6}, for each boundary superselection sector $\alpha$, there exists a state $\tilde{\sigma}^{(\alpha)}_{ABC}$, such that
	\begin{eqnarray}
	\tilde{\sigma}^{(\alpha)}_{AB} &=& {\sigma}^{\alpha}_{AB}\nonumber\\
	\tilde{\sigma}^{(\alpha)}_{BC} &=& \sigma^1_{BC}\nonumber\\
	I(A:C\vert B)_{\tilde{\sigma}^{(\alpha)}} &=& 0 \nonumber\\
	\tilde{\sigma}^{(\alpha)}_{AC}&=&\sigma^{\alpha}_{A}\otimes \sigma^1_{C}.\nonumber
	\end{eqnarray}

\begin{figure}[h]
	\centering
	\includegraphics[scale=1.0]{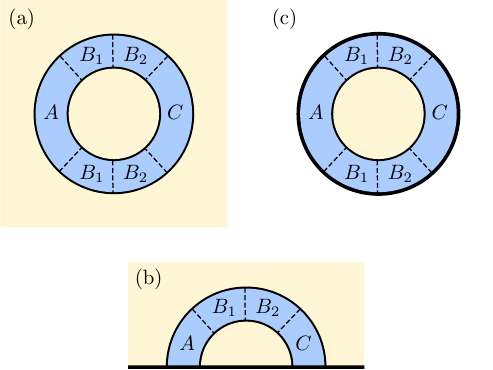}
	\caption{A subsystem $\Omega$ is divided into $ABC$, $B=B_1B_2$. (a)  $\Omega=\Omega_1$ is an annulus inside the bulk. (b)  $\Omega=\Omega_2$  attaches to a gapped boundary by two connected components. (c)  $\Omega=\Omega_3$ is an annulus covering a gapped boundary. }
	\label{B1B2}	
\end{figure}

\begin{figure}[h]
	\centering
	\includegraphics[scale=1.0]{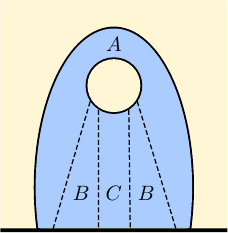}
	\caption{Subsystem $\Omega_6$ is divided into $ABC$.}
	\label{2nd_way_to_divide_Omega6}	
\end{figure}

Because $\tilde{\sigma}_{ABC}$  looks like a ground state locally, we argue that \footnote{We do not  have a proof of this argument based on the assumptions shown in this paper, but the bulk version of this argument is proved under axioms {\bf A0} and {\bf A1} in~\cite{2019arXiv190609376S}.}  it is an element in a certain information convex. Then, it has to be a certain maximal-entropy element. This leads to the following results. (We omit the details of the derivations because they are essentially the same as those shown in Sec.~\ref{Sec_Bulk_da}.)
\begin{itemize}
	\item Let $ABC=\Omega_1$ as shown in Fig.~\ref{B1B2}(a); we have
	\begin{equation}
	\tilde{\sigma}_{\Omega_1}=\sum_{a}\frac{d_a^2}{\mathcal{D}^2} \sigma^a_{\Omega_1}.
	\end{equation}
    \item  Let $ABC=\Omega_2$ as shown in Fig.~\ref{B1B2}(b); we have
   	\begin{equation}
	\tilde{\sigma}_{\Omega_2}=\sum_{\alpha}\frac{d_{\alpha}^2}{\big(\sum_{\beta} d_{\beta}^2 \big)}  \sigma^{\alpha}_{\Omega_2}.
	\end{equation}
	\item Let $ABC=\Omega_3$ as shown in Fig.~\ref{B1B2}(c); we have
	\begin{equation}
	\tilde{\sigma}_{\Omega_3}= \sum_a \frac{N_{a}^1 d_a}{(\sum_{b}N_{b}^1 d_{b})} \sigma^{(a)_{\textrm{max}}}_{\Omega_3}.
	\end{equation}
    Here, $\sigma^{(a)_{\textrm{max}}}_{\Omega_3}$ is the maximal-entropy element of $\Sigma(\Omega_3)$ which has charge $a$ on the entanglement cut.
	\item Let $ABC=\Omega_6$ as shown in Fig.~\ref{2nd_way_to_divide_Omega6}; we have
	\begin{equation}
	\tilde{\sigma}^{(\alpha)}_{\Omega_6} =\sum_a \frac{ N_a^{\alpha} d_a }{( \sum_b N_b^{\alpha} d_b )} \sigma^{(a,\alpha)_{\textrm{max}}}_{\Omega_6}.
	\end{equation}
	Here, $\sigma^{(a,\alpha)_{\textrm{max}}}_{\Omega_6}$ is the maximal-entropy element of $\Sigma_{a}^{\alpha}(\Omega_6)$.
\end{itemize}

From these, one can further show the following:
\begin{itemize}
	\item The TEE \cite{2006PhRvL..96k0404K,2006PhRvL..96k0405L} $\gamma$ can be derived  from any one of $\Omega_1$, $\Omega_2$ and $\Omega_3$:
	\begin{equation}
	\gamma = \ln \mathcal{D} =\ln (\sum_{\alpha} d_{\alpha}^2) = \ln (\sum_a N_a^1 d_a).
	\end{equation}
	with a mild assumption that the universal piece depends solely on the topology. Here $\mathcal{D}=\sqrt{\sum_a d_a^2}$ is the total quantum dimension.
	This further implies
	\begin{equation}
	\sqrt{\sum_a d_{a}^2} =\sum_{\alpha}d_{\alpha}^2 = \sum_a N_a^1 d_a .\label{eq: Sum of quantum dimension}
	\end{equation}
	
	\item The  $\Omega_6$ case in Fig.~\ref{2nd_way_to_divide_Omega6} further implies 
	\begin{equation}
	 \sum_{a} N_a^{\alpha} d_a = d_{\alpha}\mathcal{D}. \label{eq:Sum of da is d_alpha D}
	\end{equation}
\end{itemize}
	In fact, Eq. (\ref{eq: Sum of quantum dimension}) could be derived from Eq.~(\ref{eq:conservation of quantum dimension}) and Eq.~(\ref{eq:Sum of da is d_alpha D}). 	Both    Eq.~(\ref{eq:conservation of quantum dimension}) and Eq.~(\ref{eq:Sum of da is d_alpha D}) have been observed in some examples \footnote{These results are obtained as theorems in \cite{2017CMaPh.355..645C}, but the physical context is different because the reference describes confined boundary excitations.}. We show that these results emerge from information-theoretic considerations.

\section{Summary and discussion}\label{Sec:summary}
We have presented a derivation of the $\ln d_{a}$ contribution to the von Neumann entropy by looking into the information-theoretic consistency of the information convex structure. The information-theoretic consistency is from the strong subadditivity and the conditional independence in the ground state of topologically ordered systems. A certain topological way of encoding the fusion theory in the information convex is assumed. Notably, we assume that fusion multiplicities are coherently encoded in the information convex of a 2-hole disk. 
The proof generalizes to gapped boundaries. A derivation of the $\ln d_{\alpha}$ entropy contribution from boundary topological excitation $\alpha$ is presented. Also following from this method are certain constraints of the fusion theory, identifying the fusion probabilities and a linear bound on the circuit depth of non-Abelian unitary string operators.

This derivation points to a different perspective of the origin of the universal properties of a topologically ordered system, namely the consistency with quantum mechanics of the many-body system. We believe that the initial steps presented in this work could be further developed in two directions. First, the assumptions may be simplified. A new theoretical framework may emerge in this direction. Second, we expect the key method to apply to a broader physical context.  We expect that this work will extend our toolbox and provide us with a logic generalizable to 3D topological orders which have potentially linked or knotted looplike excitations.
Because gapped boundaries are special cases of gapped domain walls, it is natural to think about a generalization of this method to the study of gapped domain walls \cite{2012CMaPh.313..351K,2019PhRvB..99o5134M}. Furthermore, because the generic gapped domain walls between two topologically ordered phases are closely related to anyon condensation \cite{2014NuPhB.886..436K}, this may further allow us to study properties of anyon condensation.

Our method applies to topologically ordered systems with both Abelian and non-Abelian anyons. It provides a certain perspective on the density matrix structure of low-energy states. We notice that quantum error correction with non-Abelian anyons remains as a challenging problem \cite{2014PhRvX...4c1058B,2017CMaPh.355..519D}. It might be interesting to see whether this work provides any useful perspective on solving this problem.

The method in this paper requires the finiteness of correlation length. This is why we have to restrict to gapped boundaries. Moreover, only the fusion properties are considered. It would be desirable to find ways to go beyond these limitations so that one may study the boundary of a 2D system with a generic chiral central charge $c_{-}$ and to handle the braiding properties, e.g., the topological spins ($\theta_a$). In particular, it is an intriguing question as to whether it is possible to derive the formula
$e^{2\pi i c_{-}/8}=\mathcal{D}^{-1}\sum_{a} d_{a}^2 \theta_a$
in a suitable generalization of our method. We leave this for future studies.

\emph{Note added.} 
An important progress \cite{2019arXiv190609376S} has been made after this work is posted. In particular, a new definition of information convex is proposed which takes a single quantum state as the input. For a closed manifold, versions of \ref{as:G0}, \ref{as:G1}, \ref{as:S1}, \ref{as:1}, \ref{as:2}, \ref{as:3}, \ref{as:F1}, \ref{as:F2}, \ref{as:F3}, \ref{as:F4}, \ref{as:F5} are derived from two local entropic constraints originally proposed by Kim \cite{2014arXiv1405.0137K} (namely axioms {\bf A0} and {\bf A1} of Ref.~\cite{2019arXiv190609376S}). Therefore, the results in this paper which are derived from these assumptions hold whenever these simpler conditions hold.

\section*{Acknowledgements}
I thank Yuan-Ming Lu for helpful discussions during and after a previous project. The key step of this work was obtained by the author during PiTP 2018 ``From Qubits to Spacetime" at IAS; the lecture by Atish Dabholkar, ``Exact Quantum Entropy of
Black Holes," led the author to think about a related problem. I thank Yanjun He for a discussion on fusion probabilities.
I have benefited from discussion with and feedback from Sebas Eli\"ens, Tarun Grover, Jeongwan Haah, Timothy  Hsieh, Taylor Hughes, Ling-Yan Hung, Liang Kong, Chong Wang, Xueda Wen, Dominic Williamson, and Beni Yoshida. I thank Kohtaro Kato for explaining to me the ``merging technique" during a conference at KITP 2017.
I give special acknowledgment to Isaac~H.~Kim for recommending two entropic constraints at the beginning of the collaboration on Ref.~\cite{2019arXiv190609376S}; these two entropic constraints turned out to be the foundation of a more complete framework.  
This work is supported by the National Science Foundation under Grant No. NSF DMR-1653769.

\appendix

\section{String operators vs.  unitary string operators}\label{appendix:string}
In the literature, there are discussions of string operators and  unitary string operators. For example, Ref.~\cite{2006AnPhy.321....2K} defines an excitation with a nontrivial superselection sector to be an excitation which cannot be created by any local operator, and it is argued that topological excitations could be moved or created using operators acting on a string.  While in known exactly solvable models Abelian anyon strings are unitary (e.g., in the toric code model), there are nonunitary string operators which could create a pair of non-Abelian topological excitations when acting on a ground state, e.g., the ribbon operators in non-Abelian quantum double models \cite{2003AnPhy.303....2K,2008PhRvB..78k5421B} and the string operators in string-net models \cite{2005PhRvB..71d5110L}. 
On the other hand, some references consider unitary string operators which create a pair of generic topological excitations, e.g., \cite{2015PhRvB..92k5139K,2019PhRvB..99c5112S}. This is because  unitarity  is relevant to some information-theoretic properties. 

In this section, we make clear the relation between these two seemingly different requirements. In particular, we  prove Proposition~\ref{Prop_Unitary}, which says that on a ground state satisfying \ref{as:G0}, any  operator supported on a subsystem could be replaced by a corresponding unitary operator supported on a slightly thicker subsystem (thicker by the scale $\epsilon$ in \ref{as:G0}). 

\begin{Proposition}\label{Prop_Unitary}
	Let $X$ be an operator supported on subsystem $A$. $A_{\epsilon}$ is a subsystem thicker than $A$ by the length scale $\epsilon$ required in \ref{as:G0}. 
	 If $X\vert \psi\rangle \ne 0$, then there exists a unitary operator $U(X)$ supported on $A_{\epsilon}$ such that
	\begin{equation}
	c \cdot X\vert \psi\rangle =  U(X) \vert \psi\rangle 
	\end{equation}
	for a state $\vert \psi\rangle$ satisfying \ref{as:G0}. Here $c$ is a complex number which fixes the normalization.
\end{Proposition}
\begin{proof}
	Let $\vert \varphi(X)\rangle \equiv  c \cdot X\vert \psi\rangle $ and  $\sigma^{(X)}\equiv \vert \varphi(X)\rangle \langle \varphi(X)\vert$. It follows from \ref{as:G0} that $I(A:\bar{A_{\epsilon}})_{\sigma^1}=0$ and therefore
	$\sigma^{(X)}_{A\bar{A}_{\epsilon}} = \sigma^{(X)}_{A}\otimes \sigma^1_{\bar{A}_{\epsilon}}$,
	where $\sigma^{(X)}_{A}\equiv \vert c\vert^2 X\sigma^1_A X^{\dagger}$ is a reduced density matrix and $\bar{A_{\epsilon}}$ is the complement of $A_{\epsilon}$. In other words, $\vert \varphi(X)\rangle$ and $\vert \psi\rangle$ have the same reduced density matrix on $\bar{A}_{\epsilon}$.
	Therefore, there exists a unitary operator $U(X)$ supported on $A_{\epsilon}$ such that $
	\vert \varphi(X)\rangle =U(X) \vert \psi\rangle$. The last step follows from a general result \cite{preskill1998lecture}: if $\vert \psi_{AB}\rangle$ and $\vert \varphi_{AB}\rangle$ have the same reduced density matrix on $A$, then there is a unitary operator $U_{B}$ supported on $B$ such that $\vert \varphi_{AB}\rangle=U_{B}\vert \psi_{AB}\rangle$. This completes the proof.
\end{proof}

\section{Monotonicity of fidelity} \label{Appendix_Monotonicity}
We provide a short proof of the statement $\rho_{AB}\cdot \sigma_{AB}=0$ if $\rho_{A}\cdot \sigma_{A}=0$. In  words, two density matrices are orthogonal on $AB$ if they are orthogonal on $A$.  While this result could be proved using straightforward calculation, we provide a proof based on \emph{the monotonicity of fidelity} \cite{preskill1998lecture}. The advantage of this method is that it applies to the situation that two density matrices are approximately orthogonal.

Given any two density matrices $\rho$ and $\sigma$, the fidelity $F(\rho,\sigma)$ is defined as
\begin{equation}
F(\rho,\sigma)\equiv \bigg( \textrm{tr}\sqrt{\rho^{\frac{1}{2}}\sigma \rho^{\frac{1}{2}}} \bigg)^2.
\end{equation}
It follows that $F(\rho,\sigma)\in [0,1]$. Its value measures the overlap between two density matrices $\rho$ and $\sigma$. $F(\rho,\sigma)=0$ if and only if $\rho\cdot \sigma=0$ and $F(\rho,\sigma)=1$ if and only if $\rho=\sigma$.

The monotonicity of fidelity states that
\begin{equation}
F(\rho_{AB},\sigma_{AB})\le F(\rho_A,\sigma_A).
\end{equation}
If $\rho_A\cdot \sigma_A=0$, then $F(\rho_{AB},\sigma_{AB})\le F(\rho_A,\sigma_A)=0$. $F(\rho_{AB},\sigma_{AB})=0$ is the only choice because it cannot be negative. Therefore, $\rho_{AB}\cdot\sigma_{AB}=0$.

\section{Quantum dimensions from fusion}\label{Appendix_Fusion}

The general definition of a quantum dimension for a relatively generic fusion category (without the assumption that $N_{ij}^k=N_{ji}^k$) can be found in Appendix E of \cite{2006AnPhy.321....2K}. Note that we do need to avoid this assumption since the category describing the boundary theory does not have this property in general.  
In this appendix, we provide a self-contained proof of the existence and uniqueness of the quantum dimension from a set of assumptions listed below. This set of assumptions is general enough to apply for bulk anyons  and boundary topological excitations. In this proof, a key step is the Perron-Frobenius theorem for matrices with positive entries.

Let us first briefly review the  Perron-Frobenius theorem. Here an $N\times N$ matrix $[A]$ is called a matrix with positive entries if  we have matrix element $[A]_{ij}>0$ for any $i,j\in \{1,\cdots , N \}$.
\begin{theorem}[Perron-Frobenius]
	Let $[A]$ be an $N\times N$ matrix with positive entries; then the following statements hold:
	\begin{itemize}
		\item There is a real and positive number $r$ being an eigenvalue of $[A]$ and $\vert \lambda\vert <r$ for all other eigenvalues $\lambda$ of $[A]$. Here $\lambda$ can be complex.
		\item $r$ is simple. In other words, it corresponds to a single $1\times 1$ Jordan block.
		\item There exists a vector $\vert v\rangle =(v_1,\cdots, v_N)^T$ and $\vert w \rangle =(w_1,\cdots, w_N)^T$, with real $v_i$ and real $w_i$ $\forall i\in \{1,\cdots, N \}$, such that
		\begin{equation}
		[A]\vert v\rangle =r \,\vert v\rangle, \quad  [A]^T \vert w\rangle = r\,\vert w\rangle.
		\end{equation}
		Such $\vert v\rangle$ and $\vert w\rangle $ are unique up to rescaling. Here ``$\,T$" means transpose.
		\item $\vert v\rangle $ is the only non-negative eigenvector of $[A]$. Similarly, $\vert w\rangle $ is the only non-negative eigenvector of $[A]^T$. Here, a non-negative eigenvector is an eigenvector with  non-negative entries.
	\end{itemize}
\end{theorem}

Next, we  list some standard assumptions in the fusion theory, which are needed in order to  derive the existence and uniqueness of the quantum dimension. The last assumption is not needed for the proof of existence and uniqueness but is needed in order to show $d_j= d_{\bar{j}}$ and $d_{j}\ge 1$.
\begin{assumption}{F_1}{} \label{as:F1}
	There exists a finite set of labels $\mathcal{C}=\{ i,j,k,\cdots\}$. (This implies that we could take $i=1, \cdots, N$ for a positive integer $N$.) There exist coefficients $N_{ij}^k$ which take non-negative integer values and we call them fusion multiplicities.
\end{assumption}

\begin{assumption}{F_2}{} \label{as:F2}
	Fusion  is associative:
	\begin{equation}
	\sum_m N_{ij}^m N_{mk}^l =\sum_n N_{in}^{l} N_{jk}^n.
	\end{equation}
\end{assumption}

\begin{assumption}{F_3}{} \label{as:F3}
	There exists a unique vacuum sector $1\in \mathcal{C}$ and 
	\begin{equation}
	N_{1i}^j = N_{i1}^j =\delta_{i,j}.
	\end{equation}
\end{assumption}

\begin{assumption}{F_4}{}\label{as:F4} For each label $i\in \mathcal{C}$, there exists a unique antilabel $\bar{i}\in \mathcal{C}$
	such that:
	\begin{equation}
	N_{ij}^1 = N_{ji}^1=\delta_{j,\bar{i}}.
	\end{equation}
\end{assumption}

\begin{assumption}{F_5}{} \label{as:F5}
	\begin{equation}
	N_{ij}^k = N_{\bar{j}\bar{i}}^{\bar{k}}.
	\end{equation}
\end{assumption}
We have finished the discussions about the fusion assumptions. Note that we do not assume $N_{ij}^k=N_{ji}^k$ since the boundary theories do not satisfy this property in general. Next, we provide a few definitions and some simple corollaries which will be useful in the proof.

Let us define $[N_i]$ to be the matrix with component $[N_i]_{jk}= N_{ij}^k$. We further define $[N(p)]\equiv \sum_i p_i [N_i]$, for any probability distribution $\{ p_i\}$ with $p_i>0$, $\forall i$. The following are some corollaries:

\begin{enumerate}
	\item From \ref{as:F1}, we know that each $[N_i]$ is a matrix with non-negative entries.
	\item It follows from  \ref{as:F2} that
	\begin{eqnarray}
	\sum_k N_{ij}^k [N_k] &=& [N_j][N_i],  \label{eq:Fusion_1} \\ 
	\sum_k [N(p)]_{jk} [N_k] &=& [N_j][N(p)]. \label{eq:Fusion_2}
	\end{eqnarray}
	\item It follows from \ref{as:F3} and \ref{as:F4} that 
	\begin{equation}
	\bar{1}=1,\quad \textrm{and}\quad \bar{\bar{i}}=i.
	\end{equation}
	\item From \ref{as:F1}, \ref{as:F2}, \ref{as:F3} and \ref{as:F4}, one can show
	\begin{equation}
	\begin{aligned}
	\sum_{i} N_{ij}^k >0,&\quad \forall\, j,k\\
	\sum_{j} N_{ij}^k >0,&\quad \forall \,i,k\\
	\sum_{k} N_{ij}^k >0,&\quad \forall \, i,j 
	\end{aligned}\label{Positive_useful}
	\end{equation}
	and it follows that
	$[N(p)]$ is a matrix with positive entries. In other words, $[N(p)]_{jk}>0$, $\forall \, j,k$.
	\item It follows from \ref{as:F1}, \ref{as:F2}, \ref{as:F3} and \ref{as:F4} that 
	\begin{equation}
	N_{ij}^k = N_{j\bar{k}}^{\bar{i}}.
	\end{equation}
	\item It follows from \ref{as:F1}, \ref{as:F2}, \ref{as:F3}, \ref{as:F4} and  \ref{as:F5} that
	\begin{equation}
	N_{ij}^k= N_{\bar{i}k}^j =N_{j\bar{k}}^{\bar{i}} =N_{\bar{j}\bar{i}}^{\bar{k}} =N_{\bar{k}i}^{\bar{j}} =N_{k\bar{j}}^i, 
	\end{equation}
	and therefore
	\begin{eqnarray}
	[N_{\bar{i}}] = [N_i]^T.
	\end{eqnarray}
\end{enumerate}

\begin{Proposition} \label{prop_quantum_dim}
	There exists a positive eigenvector $\vert v\rangle =(v_1,\cdots, v_N)^T$, with $v_i>0$, $\forall i\in \{ 1,\cdots, N \}$ such that
	it is the common eigenvector of all the matrices $[N_i]$ and $[N(p)]$ for any $\{ p_i \}$. The quantum dimension $d_i>0$  can be defined from
	\begin{equation}
	[N_i]\vert v\rangle =d_i \vert v\rangle.
	\end{equation}
	The quantum dimension $d_i$ is the largest eigenvalue of $[N_i]$. Furthermore, $\{d_i \}$ is the unique set of positive numbers satisfying the following equation:
	\begin{equation}
	d_i d_j =\sum_k N_{ij}^k d_k. \label{Fusion to proof}
	\end{equation}
\end{Proposition}
The proof of \ref{prop_quantum_dim} follows from \ref{as:F1}, \ref{as:F2}, \ref{as:F3} and \ref{as:F4}. In fact, these assumptions could be further relaxed if we just want to have a unique positive solution of Eq. (\ref{Fusion to proof}).
\begin{proof}
	Since $[N(p)]$ is a matrix with positive entries, according to  the Perron-Frobenius theorem, there exists a unique positive vector $\vert v\rangle =(v_1,\cdots, v_N)^T$ such that 
	\begin{equation}
	[N(p)] \vert v\rangle =r_{p} \vert v\rangle.
	\end{equation}
	Then by using Eq. (\ref{eq:Fusion_2}) we get
	\begin{equation}
	\sum_k [N(p)]_{jk} [N_k] \vert v\rangle =r_p [N_j] \vert v\rangle.   \label{C15}
	\end{equation}
	Define $\vert v(j)\rangle \equiv [N_j]\vert v\rangle$. Then $\vert v(1)\rangle =\vert v\rangle $ since $[N_1]=1$ is the identity matrix (assumption \ref{as:F3}), and $\vert v(j)\rangle $ is positive because of Eq. (\ref{Positive_useful}) and that $\vert v\rangle $ has only positive entries. Therefore, we can rewrite Eq. (\ref{C15}) into
	\begin{equation}
	\sum_{k} [N(p)]_{jk} \vert v(k)\rangle = r_p \vert v(j)\rangle.
	\end{equation}
	Then it is straightforward to show that $\vert v(j)\rangle \propto \vert v\rangle$ and the proportional coefficient has to be positive; let us call it $d_j$ and therefore
	\begin{equation}
	[N_j]\vert v\rangle =d_j \vert v\rangle,\quad \forall\, j.
	\end{equation}
	It shows that $\vert v\rangle $ is a common eigenvector of all $[N_j]$ and also $[N(p)]$ for an arbitrary probability distribution $\{ p_i \}$.  In fact, $d_j$ has is the largest eigenvalue since $\vert v\rangle $ is positive and $[N_j]$ could be obtained by taking a certain limit of $[N(p)]$.
	
	Next, we use Eq. (\ref{eq:Fusion_1}) to show that $d_i d_j =\sum_k N_{ij}^k d_k$.
	Thus, we have proved that there is at least one positive solution to Eq. (\ref{Fusion to proof}).
	Now let us show that the positive solution of Eq. (\ref{Fusion to proof}) is unique. This follows from the fact that if Eq.~(\ref{Fusion to proof}) holds, and $d_j$ is positive, then $d_j$ has to be the largest eigenvalue of $[N_j]$ and such largest eigenvalue is unique, given the choice of $N_{ij}^k$.
	
	We have finished the proof using assumptions \ref{as:F1}, \ref{as:F2}, \ref{as:F3}, \ref{as:F4}.
\end{proof}
With this proposition, one could define quantum dimension $d_i$ as the unique positive solution of Eq. (\ref{Fusion to proof}). Given that the assumptions apply to both the bulk and the boundary, we conclude that the bulk quantum dimension $d_a$ can be defined as the unique positive solution of $d_a d_b =\sum_c N_{ab}^c {d_c}$, and the boundary quantum dimension can be  defined as the unique positive solution of $d_{\alpha} d_{\beta}= \sum_{\gamma} N_{\alpha\beta}^{\gamma} d_{\gamma}$.

\begin{Proposition}
	\begin{equation}
	d_1=1.
	\end{equation}
\end{Proposition}
\begin{proof}
	Let us apply  Eq. (\ref{Fusion to proof}) to the vacuum sector
	\begin{equation}
	\begin{aligned}
	d_1 d_{{1}} &=\sum_k N_{1 1}^k d_k \\
	&= N_{1 1}^1 d_1\\
	&=d_1.
	\end{aligned}
	\end{equation}
	Thus, $d_1=1$. 	We have finished the proof using assumptions \ref{as:F1}, \ref{as:F2}, \ref{as:F3}, \ref{as:F4}.
\end{proof}

\begin{Proposition} \label{da=dabar}
	\begin{equation}
	d_i=d_{\bar{i}}.
	\end{equation}
\end{Proposition}
\begin{proof}
	It follows from two facts:
	\begin{enumerate}
		\item The quantum dimension $d_i$ is the largest positive eigenvalue of $[N_i]$, and $d_{\bar{i}}$ is the largest positive eigenvalue of $[N_{\bar{i}}]$.
		\item  $[N_{\bar{i}}]= [N_i]^T$ and therefore they have the same largest positive eigenvalue. In this step,  \ref{as:F5} is used.
	\end{enumerate}
	We have finished the proof using assumptions \ref{as:F1}, \ref{as:F2}, \ref{as:F3}, \ref{as:F4}, \ref{as:F5}.
\end{proof}

\begin{Proposition} 
	\begin{equation}
	d_i\ge 1,\qquad \forall \,i.
	\end{equation}
	and $d_i>1$ if and only if $i$ and $\bar{i}$ have a fusion channel other than the vacuum $1$.
\end{Proposition}
\begin{proof}
	It follows from Eq. (\ref{Fusion to proof}) that
	\begin{equation}
	\begin{aligned}
	d_i d_{\bar{i}} &=\sum_k N_{i\bar{i}}^k d_k \\
	&\ge N_{i\bar{i}}^1 d_1 \\
	&=1.
	\end{aligned}
	\end{equation}
	Then applying Proposition~\ref{da=dabar}, we get $d_i^2\ge 1$ and therefore the positive number $d_i\ge 1$. One could further observe that if $1$ is the only fusion channel of $i$ and $\bar{i}$, then $d_i=1$, and if there are other fusion channels, we have $d_i>1$.
	
	We have finished the proof using assumptions \ref{as:F1}, \ref{as:F2}, \ref{as:F3}, \ref{as:F4}, \ref{as:F5}.
\end{proof}

\bibliography{ref}
\bibliographystyle{apsrev}
\end{document}